\documentclass[11pt]{article}

\usepackage{amsmath}
\usepackage{amsthm}
\usepackage{amssymb}
\usepackage{fancyhdr}
\usepackage{hyperref}

\newtheorem{definition}{Definition}
\newtheorem{lemma}{Lemma}
\newtheorem{theorem}{Theorem}
\newtheorem{corollary}{Corollary}

\theoremstyle{remark}
\newtheorem{remark}{Remark}

\title{Six Birds:\\ Foundations of Emergence Calculus}
\author{Ioannis Tsiokos\\\texttt{ioannis@automorph.io}}
\date{25 January 2026}

\fancypagestyle{firstpage}{%
  \fancyhf{}
  \fancyfoot[C]{\scriptsize DOI: \href{https://doi.org/10.5281/zenodo.18365949}{10.5281/zenodo.18365949}}

}

\begin{document}
\maketitle
\thispagestyle{firstpage}
\vspace{-0.5em}
\begin{center}
\scriptsize \textcopyright\ 2026 Automorph Inc.
\end{center}
\vspace{-0.75em}

\begin{abstract}
We present a math-only framework in which \emph{theories} arise as fixed points of idempotent operators and \emph{open-endedness} requires changing the operator rather than iterating it.
Under minimal hypotheses of composability plus limited interface access with bounded observation, we show that the operator vocabulary P1–P6 arises canonically as unavoidable closure mechanics (Theorem~\ref{thm:meta-prim}).
Order-theoretic closure operators and dynamics-induced endomaps $E_{\tau,f}$ (built from a Markov kernel, a coarse-graining lens, and a timescale) are treated uniformly via the abstraction of idempotent endomaps and their fixed points.
We define a computable total-variation idempotence defect for $E_{\tau,f}$ and show that small retention error implies approximate idempotence, yielding robust packaged “objects” at a given time-scale $\tau$ within a fixed lens/theory.
To prevent spurious directionality, we define arrow-of-time as path-space KL divergence and prove that coarse-graining cannot increase it (data processing); we further formalize the protocol-trap principle showing that autonomous protocol holonomy does not generate sustained asymmetry unless a genuine affinity is present in the lifted dynamics.
Finally, we give a finite forcing-style counting lemma: definable predicates are exponentially rare relative to a partition-based theory, so generic predicate extensions are overwhelmingly non-definable, providing a mathematically clean anti-saturation mechanism for strict ladder climbing.
\end{abstract}

\clearpage

\tableofcontents

% Scope IDs: see docs/spec/theorem-inventory.md and paper-contract.md
% Main-text claims include: T-CL-01, T-AOT-01, T-AOT-02, T-ACC-01, T-IC-01, T-IC-02, T-P1-01, T-P2-01, T-FOR-01, T-FOR-02.
\section{Introduction}
This paper isolates a finite, discipline-agnostic \emph{emergence calculus}: a minimal operator toolkit for when a coarse description becomes a \emph{theory} with stable objects, how strict theory extension remains possible, and how directionality audits stay honest under observation.
A new structural result makes this toolkit non-ad hoc: Theorem~\ref{thm:meta-prim} shows that once processes are composable but interface-limited, the operator vocabulary P1--P6 arises canonically as unavoidable closure mechanics of description (we refer to these as the “six birds”).
The organizing principle is to distinguish \emph{completion} (idempotent endomaps with fixed points) from \emph{extension} (refinements of what is definable), and to keep directionality certificates strictly tied to audit functionals that are monotone under coarse observation.
Within that scope we prove a small spine of claims: closure saturation, data processing for path reversal asymmetry, a protocol-trap audit for hidden schedules, and a finite forcing-count lemma that quantifies definability rarity.

The manuscript is structured to keep assumptions explicit and finite.
After conventions, we develop order-closure and idempotent endomaps, connect them to dynamics-induced
operators on the simplex, and then turn to accounting and arrow-of-time diagnostics.
We treat definability/forcing in finite theory terms and collect primitives P1--P6 as closure-changing
operations. The appendices record reproducibility checks, a Lean map, a Python evidence map, a Zeno
decision frontier, and a toolkit appendix that packages defect calculus, route mismatch, and bridge
objects for reuse.
A compact interaction diagram for how the six primitives fit into the rest of the calculus appears in Section~\ref{sec:six-birds-loop}.
We then make explicit how the same toolkit expresses \emph{downward influence} (macro-to-micro constraint within a coupled theory package) in Section~\ref{sec:downward-influence}.

\begin{center}
\fbox{\begin{minipage}{0.96\linewidth}
\begin{remark}[Interoperability of certificates (toolkit safety)]\label{rem:tk-interoperability}
This manuscript uses three logically distinct ``certificate families'':
\emph{stability} (existence/robustness of fixed points of a completion rule),
\emph{novelty} (strict theory extension by adding non-definable predicates),
and \emph{directionality} (nontrivial audit values such as affinities or path reversal asymmetry).
They are intentionally separated to prevent accidental overinterpretation.
The three certificate families (Stability, Novelty, Directionality) are logically independent. Each can hold or fail independently of the others. The only designed interactions are: (C1) completion defines stable objects, (C2) extension changes definability, (C3) audit is monotone under coarse-graining. No other implications are intended or valid.

\smallskip
\noindent\textbf{Non-implications.}
\begin{itemize}
  \item Stability (small idempotence defect / approximate closure) does \emph{not} imply novelty:
  a system may stabilize within a fixed definability structure (Section~\ref{sec:idempotent-endo}) without any strict theory extension (Section~\ref{sec:forcing}).
  \item Novelty does \emph{not} imply directionality:
  adding a generic predicate (Section~\ref{sec:forcing}) can change what is expressible while leaving the audit certificates unchanged.
  \item Stability and novelty do \emph{not} imply directionality:
  directionality is certified only by the audit functional and its monotonicity laws (Sections~\ref{sec:acc} and \ref{sec:aot}).
\end{itemize}

\noindent\textbf{No false positives for directionality under observation.}
Coarse-graining cannot create arrow-of-time: the path reversal asymmetry contracts under pushforward (data processing in Section~\ref{sec:aot}).
Apparent irreversibility caused by hiding an internal protocol/phase is diagnosed by the protocol-trap/clock-audit pattern (Section~\ref{sec:aot}).

\smallskip
\noindent\textbf{Where coupling actually occurs.}
The only intended coupling mechanisms are explicit operators:
(i) \emph{completion} changes fixed-point sets (Sections~\ref{sec:closure} and \ref{sec:idempotent-endo}),
(ii) \emph{extension} refines definability (Section~\ref{sec:forcing}),
and (iii) the audit functional is functorial/monotone under lenses (Sections~\ref{sec:acc} and \ref{sec:aot}).
Appendix~E's ECT template is an optional coupling that links completion/accounting to a throughput/capacity certificate under explicit slots; it still does not certify novelty or directionality.
Appendix~E's balanced-atom route (Definitions~\ref{def:ect-atom-ss}--\ref{def:ect-bal} and Corollary~\ref{cor:ect-bal-icap}) is one optional way to couple packaging/accounting to a throughput certificate; it still does not certify novelty (Section~\ref{sec:forcing}) or directionality (Sections~\ref{sec:acc}--\ref{sec:aot}).
Route dependence (holonomy) is treated separately as a reduction-order effect and is not, by itself, a directionality certificate.
\end{remark}
\end{minipage}}
\end{center}

\subsection{The organizing picture: a three-certificate loop}\label{sec:big-picture}

The results of this paper can be read as a minimal calculus for building hierarchical structure while keeping directionality honest.
The core observation is that three notions often conflated in informal discussions are logically distinct and have distinct mathematical \emph{certificates}:

\begin{itemize}
  \item \textbf{Stability / objects (completion).} A theory level is specified by an idempotent endomap (exact or approximate); its fixed points are the “objects” at that theory level.
  In order theory this is a closure operator; in dynamics it is the induced endomap $E_{\tau,f}$.
  The certificate is (approximate) idempotence, quantified by the TV defect $\delta_{\tau,f}$.
  \item \textbf{Novelty / open-endedness (theory extension).} Iterating a fixed completion rule saturates by idempotence.
  Strict ladder climbing therefore requires changing the theory/completion rule.
  The certificate is \emph{non-definability}: a new predicate is not measurable with respect to the old partition-based theory.
  Our finite forcing lemma makes this generic: when there is hidden volume ($N-K>0$), uniformly random predicate extensions are overwhelmingly non-definable.
  \item \textbf{Directionality / irreversibility (audit).} Arrow-of-time is defined on path space as $\Sigma_T(\rho)=D_{\mathrm{KL}}(\mathbb P_{\rho,T}\|\mathcal{R}_{*}\mathbb P_{\rho,T})$.
  The certificate is either positive path-space KL or, under the accounted regime, a non-exact log-ratio 1-form (nonzero cycle affinities).
  Data processing ensures coarse-graining cannot create this certificate, and the protocol-trap theorem ensures it is not an artifact of hiding an internal clock.
\end{itemize}

These certificates compose into a simple loop:
choose a theory (lens) $f$, induce a packaging endomap $E_{\tau,f}$ from dynamics and timescale, extract its (approximate) fixed points as objects, then adjoin a generic non-definable predicate to obtain a strict theory extension and repeat—while auditing any apparent arrow-of-time via path-space KL and affinity.
The main text develops each piece of this loop in a representation-independent way and keeps all nonessential computational evidence in the appendices.
\paragraph{Executive map (emergence calculus).}
A \emph{theory} in this paper is a package of (i) a lens/definability structure, (ii) a completion rule, and (iii) an audit.
The six primitives are the minimal moves that act on these components:
P5 selects objects as fixed points of a completion,
P6 supplies audits/monotones that remain honest under observation,
P4 provides a bounded theory index for coherent refinement,
P2 carves feasibility and representability,
P1 rewrites dynamics when closure fails to descend.
P3 diagnoses non\-commuting reduction routes (route mismatch); this is \emph{not} a directionality certificate.
Emergence means the appearance of stable fixed points within a theory.
Open-endedness means strict theory extension rather than iteration of a fixed closure.
Section~\ref{sec:six-birds-loop} makes this interaction loop explicit.

\paragraph{Roadmap.}
Section~\ref{sec:closure} develops the order-closure backbone and the saturation principle.
Section~\ref{sec:idempotent-endo} isolates idempotent endomaps as the unifying abstraction and introduces the induced empirical endomap $E_{\tau,f}$ with a computable idempotence defect and a qualified theory-birth criterion.
Section~\ref{sec:acc} formulates \textbf{A\_AUT}+\textbf{A\_REV}+\textbf{A\_ACC} as a graph 1-form exactness statement and defines affinities.
Section~\ref{sec:aot} defines arrow-of-time as path reversal asymmetry, proves data processing, and formalizes the protocol-trap/clock-audit principle (including the corrected autonomous reading of protocol holonomy).
Section~\ref{sec:forcing} gives the finite forcing lemma as an anti-saturation move for strict theory extension.
Section~\ref{sec:primitives} summarizes the primitives P1--P6 as closure-changing operations, and Section~\ref{sec:examples} gives a few surgical examples.
Section~\ref{sec:six-birds-loop} ties the primitives into the closure/forcing/audit spine as an explicit emergence calculus.

\section{Related work}\label{sec:related}

This paper is written to be largely self-contained; we briefly situate standard ingredients that
we use as mathematical tools.

\paragraph{Closure operators, reflections, and idempotents.}
Closure operators on posets/lattices are classical objects in order theory: extensive, monotone,
and idempotent endomaps whose fixed points form a ``closed'' substructure.
In categorical language, (order-theoretic) closure operators are the thin-category shadow of
reflective subcategories and idempotent monads: a reflector is left adjoint to the inclusion of the
full subcategory of fixed points, and the universal property characterizes the ``least closed object
above'' a given one.
Our main text uses the order-theoretic formulation, and the Lean scaffold mechanizes the thin
(poset) instance. See, e.g., \cite{DaveyPriestley2002,MacLane1998,Borceux1994}.

\paragraph{Coarse-graining of Markov dynamics and lumpability.}
State aggregation and lumpability for finite Markov chains provide exact hypotheses under which a
coarse observation of a Markov chain is itself Markov (or approximately so after subsampling).
We use coarse-graining primarily as a measurable map (a ``lens'') and emphasize statements that are
valid without assuming exact lumpability (notably, data-processing for relative entropy).
When we do speak about induced macro-dynamics, the relevant classical references are the standard
Markov chain texts that treat aggregation/lumpability and time reversal; see
\cite{KemenySnell1960,Norris1997,RogersPitman1981}. For approximate lumpability and metastable coarse objects in
finite-state settings, see e.g.\ \cite{KhudaBukhshEtAl2019,NilssonJacobi2010}.

\paragraph{Path-space KL and DPI.}
Relative entropy/KL divergence and the data processing inequality are standard in information theory:
pushforward under a measurable map cannot increase KL. We apply this at the level of path measures,
so that coarse observation cannot create time-reversal asymmetry.
The use of path-space relative entropy between forward and time-reversed Markov path measures as a
measure of irreversibility also appears in the literature on reversible Markov chains and (more
recently) stochastic thermodynamics; we cite these only as orientation since our arguments are finite
and measure-theoretic. See, e.g., \cite{CoverThomas2006,CsiszarKorner2011,Kelly1979,Seifert2012,Merhav2010,Puglisi2010}
and recent coarse-graining/hidden-drive analyses such as \cite{DeguntherVanderMeerSeifert2023,BauerSeifertVanderMeer2025}.

\paragraph{Graph cycles, affinities, and nonequilibrium network structure.}
The cycle/cocycle decomposition viewpoint for Markov networks and the role of affinities as
cycle integrals are standard in nonequilibrium statistical mechanics; see, e.g.,
\cite{Polettini2012,AltanerEtAl2012,DalCengioLecomtePolettini2023}.

\paragraph{Protocol geometry and stochastic pumps.}
Geometric-phase currents and stochastic-pump effects in driven Markov systems provide a related
route-dependence perspective on nonequilibrium protocols; see \cite{SinitsynNemenman2007}.

\paragraph{Definability rarity and the forcing analogy.}
In our finite setting, ``definable from a theory'' is treated as measurability with respect to a
finite $\sigma$-algebra (equivalently: being constant on the blocks of a partition). Our use of
``theory'' is intentionally modest: it is a finite definability structure (partition / finite
$\sigma$-algebra), not an infinitary logical theory. Our ``finite forcing lemma'' is a counting
statement: when blocks have nontrivial internal volume, predicates that are definable from the coarse
theory form an exponentially small subset of all predicates, so a random predicate extension is
generically non-definable.
The term ``forcing'' is used only as an analogy to the role of generic extensions in set theory; our
results are finite and do not invoke infinitary semantics. Background references include standard
finite model theory and set-theory texts, e.g.\ \cite{Libkin2004,Cohen1963,Kunen2011,MacphersonSteinhorn2011,Hyttinen2015}.

% Scope IDs: see docs/spec/theorem-inventory.md and paper-contract.md
\section{Conventions and notation}
\subsection{Finite state spaces, distributions, and kernels}

Throughout, $Z$ denotes a finite set (the microstate space). We write $\Delta(Z)$ for the probability simplex on $Z$.
We represent distributions as row vectors $\mu \in \Delta(Z)$ and Markov kernels as row-stochastic matrices $P$ so that the time-$1$ update is $\mu \mapsto \mu P$.
All logarithms are natural logs.
We use \emph{theory} (and \emph{theory level}) as the primary term throughout.

\subsection{Paths, time reversal, and relative entropy}

Fix a horizon $T \in \mathbb{N}$. A path is $z_{0:T} := (z_0,z_1,\dots,z_T)\in Z^{T+1}$.
Given an initial distribution $\rho\in\Delta(Z)$ and a time-homogeneous kernel $P$, the forward path law is
\[
\mathbb{P}_{\rho,T}(z_{0:T}) \;=\; \rho(z_0)\prod_{t=0}^{T-1} P(z_t,z_{t+1}).
\]
Let $\mathcal R:Z^{T+1}\to Z^{T+1}$ be path reversal, $\mathcal R(z_0,\dots,z_T)=(z_T,\dots,z_0)$, and let $\mathcal{R}_{*}\mathbb{P}_{\rho,T}$ denote the pushforward law.

For probability measures $p,q$ on a finite set, we use the KL divergence
\[
D_{\mathrm{KL}}(p\|q) := \sum_x p(x)\log\frac{p(x)}{q(x)},
\]
with the conventions $0\log(0/q):=0$ and $p(x)>0,\,q(x)=0 \Rightarrow D_{\mathrm{KL}}(p\|q)=+\infty$.

The finite-horizon path reversal asymmetry $\Sigma_T(\rho)$ is defined in Section~\ref{sec:aot} as the
KL divergence between $\mathbb{P}_{\rho,T}$ and its time-reversed pushforward $\mathcal{R}_{*}\mathbb{P}_{\rho,T}$.

\begin{remark}[Stationarity is not needed for $\Sigma_T$ or data processing]
The definition in Section~\ref{sec:aot} makes sense for any initial distribution $\rho$ (it may be $+\infty$ if reverse-support is missing).
The data processing inequality for KL is purely measure-theoretic and therefore applies to $\Sigma_T(\rho)$ without any stationarity assumption.
Stationarity enters only when one wants a steady-state interpretation (e.g.\ relating $\Sigma_T(\pi)$ to entropy production and/or considering long-time rates).
\end{remark}

\subsection{Coarse-graining lenses}

A (deterministic) coarse-graining or \emph{lens} is a map $f:Z\to X$ to a finite set $X$.
The pushforward of $\mu\in\Delta(Z)$ is $(f_*\mu)(x):=\sum_{z:f(z)=x}\mu(z)$.
On path space we use the coordinatewise lens $f^{T+1}:Z^{T+1}\to X^{T+1}$, and write $f_*\mathbb{P}_{\rho,T}$ for the induced macro path law.

\begin{remark}
WLOG, we replace $X$ by the image $f(Z)$ so that $f$ is surjective and every fiber $B_x=f^{-1}(x)$ is nonempty.
This does not change the induced $\sigma$-algebra/partition or any pushforward laws.
We will often refer to a lens $f:Z\to X$ (equivalently, the induced finite partition / $\sigma$-algebra)
as a \emph{finite observational theory} on $Z$: it specifies which predicates are expressible and which
distinctions are invisible. Refining $f$ by adjoining a predicate $h$ will be called a \emph{theory extension}.
\end{remark}

\subsection{A unified theory package viewpoint}\label{sec:tk-theory-package}

\begin{remark}[Emergence and open-endedness (terminology)]
Throughout, \emph{emergence} means the appearance of stable fixed points of a completion rule within a theory package,
while \emph{open-endedness} means strict theory extension (a change in definability or completion), not iteration of a fixed closure.
These are structural notions; no claim is made that either occurs in all systems.
\end{remark}

% ID: D-TK-THY-01
\begin{definition}[Finite theory package]\label{def:tk-theory-package}
Assume \textbf{A\_FIN}. A \emph{finite theory package} is a tuple
\[
\mathcal T \;=\; (Z,\ f,\ \Sigma_f,\ E,\ \mathcal A)
\]
where:
\begin{itemize}
  \item $Z$ is a finite carrier set (microstate space).
  \item $f:Z\to X$ is a \emph{lens} (coarse description). The induced definability structure $\Sigma_f$ is the finite $\sigma$-algebra (equivalently, partition) on $Z$ generated by $f$, i.e.\ the family of predicates measurable with respect to $f$.
  \item $E:\mathcal V\to\mathcal V$ is a \emph{completion/packaging endomap} on a chosen space of descriptions $\mathcal V$ associated to $Z$.
  Typically $E$ is idempotent ($E\circ E=E$) or approximately idempotent in a specified norm.
  In the empirical/dynamical instantiation $\mathcal V=\Delta(Z)$, while in the order-theoretic instantiation $\mathcal V=Z$ (with $E$ a closure operator).
  \item $\mathcal A$ is an \emph{audit functional} (a certificate of directionality or drive) defined on the relevant objects in context (e.g.\ path laws for a Markov dynamics, or antisymmetric edge 1-forms on a support graph), together with a designated monotonicity principle under coarse observation: for any further coarsening $g:X\to Y$, the induced pushforward should not increase the audit, schematically
  \[
  \mathcal A\big((g\circ f)_\#(\cdot)\big)\ \le\ \mathcal A\big(f_\#(\cdot)\big),
  \]
  whenever both sides are defined.
\end{itemize}
We call $\Sigma_f$ the \emph{expressive content} of the theory and $E$ its \emph{completion rule}; the fixed points of $E$ are the internally complete descriptions/objects recognized by the theory at the chosen scale and lens.
\end{definition}

\begin{remark}
We use the symbol $\Sigma$ in two unrelated roles: $\Sigma_f$ denotes the definability $\sigma$-algebra induced by a lens $f$, while $\Sigma_T(\rho)$ denotes the finite-horizon path reversal asymmetry defined in Section~\ref{sec:aot}.
Context (and subscripts) distinguish these usages.
\end{remark}

\begin{remark}[Instantiations in this manuscript]\label{rem:tk-theory-instantiations}
The components of $\mathcal T$ appear throughout the paper in specialized forms.
In Section~\ref{sec:closure}, an order-closure operator $c$ on a poset is an idempotent completion rule; fixed points of $c$ are the ``objects'' at that rung of a closure ladder.
In Section~\ref{sec:idempotent-endo}, we isolate idempotent endomaps as the minimal abstraction behind this fixed-point viewpoint.
In Section~\ref{sec:idempotent-endo} we also construct, from a Markov kernel, a lens $f$, and a timescale, a dynamics-induced completion endomap $E_{\tau,f}$ on $\Delta(Z)$.
In Section~\ref{sec:forcing}, adjoining a new predicate corresponds to refining $\Sigma_f$ to a strictly stronger definability structure, i.e.\ a theory extension.
Finally, Sections~\ref{sec:acc} and \ref{sec:aot} provide two audit families used here: graph 1-form exactness/affinities (ACC) and path reversal asymmetry (arrow-of-time), the latter enjoying a precise monotonicity principle under coarse-graining (data processing).
\end{remark}

\subsection{A minimal instantiation recipe}\label{sec:instantiation-recipe}
To apply the emergence calculus to an unrelated multiscale system, one supplies (at minimum) the following five pieces of data.
\begin{enumerate}
  \item \textbf{Micro description space.} A finite (or otherwise well-posed) space $Z$ of microstates, histories, or happenings, together with the ambient space of descriptions $\mathcal V$ used in the domain (e.g.\ $Z$ itself with an order, or $\Delta(Z)$ for distributional descriptions).
  \item \textbf{Lens/refinement family.} A lens $f:Z\to X$ (or a refinement family $f_j:Z\to X_j$) specifying which distinctions are observable/definable at the current theory. This induces the definability structure (written $\Sigma_f$ in the finite setting).
  \item \textbf{Completion/packaging rule.} A completion endomap $E:\mathcal V\to\mathcal V$ (often idempotent or approximately idempotent) whose fixed points are the theory's \emph{objects} (stable packaged descriptions).
  \item \textbf{Audit functional.} An audit $\mathcal A$ that is functorial/monotone under lenses (so it does not generate false positives under coarse observation). Examples in this paper include the ACC $1$-form affinities (Section~\ref{sec:acc}) and finite-horizon path reversal asymmetry (Section~\ref{sec:aot}).
  \item \textbf{Extension/update move.} A rule that changes the theory package (typically by extending definability, changing the completion rule, or rewriting admissible dynamics). Open-endedness is never obtained by iterating a fixed completion; it requires such strict extension (Section~\ref{sec:forcing}).
\end{enumerate}

Given these inputs, the three core certificates used throughout the paper can be read as pass/fail tests:
\begin{description}
  \item[\textbf{Emergence (stability/objecthood).}] Pass when the completion has multiple robust fixed points (objects) and the idempotence defect is small at the intended scale (Section~\ref{sec:idempotent-endo} and the empirical endomap construction).
  \item[\textbf{Open-endedness (novelty).}] Pass when a theory extension is \emph{strict} in the definability sense (new predicates are not definable from the old lens), so that closing again can yield genuinely new objects; the finite forcing lemma quantifies that such strict extensions are generic (Section~\ref{sec:forcing}).
  \item[\textbf{Directionality (audit).}] Pass when the chosen audit is nontrivial in the full model (e.g.\ nonzero affinity class or nonzero path reversal asymmetry) and remains consistent under observation (data processing). Apparent directionality caused by hiding protocol/state variables is ruled out by the protocol-trap analysis (Section~\ref{sec:aot}).
\end{description}

Operationally, one applies the toolkit by iterating:
choose $(f,E,\mathcal A)$ for a candidate theory, identify its objects (fixed points), audit consistency under coarse observation, and then perform an explicit extension/rewrite step to obtain the next theory package.
Appendices collect reproducibility artifacts and sanity checks for the finite Markov and finite definability instantiations used as running examples.

\subsection{Support graphs and discrete 1-forms}

Given a kernel $P$ on $Z$, its directed support has an edge $z\to z'$ whenever $P(z,z')>0$.
When microreversibility holds (support is bidirected), we also use the induced undirected support graph with edge $\{z,z'\}$ whenever $P(z,z')>0$ and $P(z',z)>0$.

On bidirected support edges we use the antisymmetric edge 1-form
\[
a(z,z') \;:=\; \log\frac{P(z,z')}{P(z',z)}.
\]
For a directed cycle $\gamma=(z_0\!\to z_1\!\to\cdots\to z_m=z_0)$ in the bidirected support, its \emph{cycle integral} is
\[
\oint_\gamma a \;:=\; \sum_{i=0}^{m-1} a(z_i,z_{i+1}).
\]
If $a$ is exact, i.e.\ there exists $\Phi:Z\to\mathbb{R}$ with $a(z,z')=\Phi(z')-\Phi(z)$ on all bidirected edges, then $\oint_\gamma a=0$ for every cycle $\gamma$.

\subsection{Assumption bundles}

We will cite standing hypotheses by short tags (e.g.\ \textbf{A\_FIN}+\textbf{A\_LENS}) rather than repeating prose.
Use the following as the canonical assumption bundle list in this manuscript (tag spellings are aligned with \path{docs/spec/theorem-inventory.md}):

\begin{description}
  \item[\textbf{A\_FIN}] Finite combinatorial setting: underlying state/index sets are finite (so sums/KL on finite alphabets are unambiguous), and when we treat completion/packaging maps as linear operators they are finite-dimensional matrices.
  \item[\textbf{A\_AUT}] Autonomy: time-homogeneous dynamics; phase variables are part of the state (no external schedule).
  \item[\textbf{A\_REV}] Microreversibility on support: $P(z,z')>0 \Rightarrow P(z',z)>0$ (bidirected support).
  \item[\textbf{A\_ACC}] Accounting/1-form decomposition: the antisymmetric log-ratio 1-form admits a decomposition into an exact part plus fixed affinity components (graph $H^1$ viewpoint).
  \item[\textbf{A\_NULL}] Null regime: all affinity coordinates / cycle integrals vanish (e.g.\ all basis affinities $A_r=0$).
  \item[\textbf{A\_STAT}] A stationary distribution $\pi$ exists (full support when required by context).
  \item[\textbf{A\_LENS}] A specified coarse-graining lens $f:Z\to X$ is part of the statement.
\end{description}

% Scope IDs: see docs/spec/theorem-inventory.md and paper-contract.md
\section{Order-closure and closure ladders}
\label{sec:closure}
\subsection{Order-theoretic closure and fixed points}

We begin with the minimal order-theoretic notion of closure. This section is intentionally
agnostic about dynamics, probability, or measurement; those enter later. The only idea we
need is that a closure operator defines a notion of ``completion,'' and its fixed points
constitute the ``objects'' of the theory determined by that completion rule.

\begin{definition}[Closure operator]\label{def:closure-operator}
Let $(L,\le)$ be a poset. A map $c:L\to L$ is a \emph{closure operator} if for all $x,y\in L$:
\begin{enumerate}
  \item (Extensive) $x \le c(x)$.
  \item (Monotone) $x \le y \Rightarrow c(x)\le c(y)$.
  \item (Idempotent) $c(c(x)) = c(x)$.
\end{enumerate}
\end{definition}

For a preorder, the same discussion applies after quotienting by the induced equivalence relation
$x\sim y \iff x\le y\le x$. We use posets to identify fixed points by equality.

\begin{definition}[Closed points / objects]\label{def:closed-points}
Given a closure operator $c:L\to L$, its \emph{closed points} (or \emph{objects of the theory determined by $c$}) are
the fixed points
\[
\mathrm{Fix}(c) \;:=\; \{x\in L : c(x)=x\}.
\]
\end{definition}

The order on $L$ induces a natural order on closure operators: we write $c \le d$ if
$c(x)\le d(x)$ for all $x\in L$ (pointwise order). Intuitively, $d$ is ``stronger'' than $c$ if it
closes more aggressively.

\begin{lemma}[Closed points are antitone in closure strength]\label{lem:closed-antitone}
If $c \le d$ are closure operators on $L$, then $\mathrm{Fix}(d)\subseteq \mathrm{Fix}(c)$.
Equivalently: if $d(x)=x$ then $c(x)=x$.
\end{lemma}

\begin{proof}
Let $x\in \mathrm{Fix}(d)$, so $d(x)=x$. By extensiveness, $x\le c(x)$. By $c\le d$ we have
$c(x)\le d(x)=x$. Hence $x\le c(x)\le x$, so $c(x)=x$.
\end{proof}

\subsection{Closure ladders and saturation}

\paragraph{The Box is the Thing.}
Idempotence implies that repeatedly applying a \emph{fixed} completion rule cannot yield unbounded novelty.

\begin{lemma}[One-step stabilization under iteration]\label{lem:closure-iterate-stabilizes}
Let $c$ be a closure operator on $L$. Define the iterate $c^{(n)}$ by $c^{(0)}=\mathrm{id}$ and
$c^{(n+1)} = c\circ c^{(n)}$. Then for all $x\in L$ and all integers $n\ge 1$,
\[
c^{(n)}(x) = c(x).
\]
\end{lemma}

\begin{proof}
For $n=1$ the statement is trivial. If $c^{(n)}(x)=c(x)$ for some $n\ge 1$, then
$c^{(n+1)}(x)=c(c^{(n)}(x))=c(c(x))=c(x)$ by idempotence.
\end{proof}

\begin{corollary}[Closure saturates]\label{cor:closure-saturates}
Fix a closure operator $c$. For any $x\in L$, the sequence
$x,\; c(x),\; c(c(x)),\; c(c(c(x))),\dots$ stabilizes after one closure step:
it is constant from the second term onward.
\end{corollary}

\begin{remark}
Viewed abstractly, a closure operator acts like a completion rule.\\
Iterating a fixed completion rule stabilizes. Thus strict ladder growth requires changing the completion rule
(equivalently, passing to a stronger theory).
\end{remark}

Thus, any notion of open-endedness based purely on repeated application of a single completion rule is
structurally limited. To obtain a strict hierarchy of ``new objects,'' one must change the completion rule.

\begin{definition}[Strictly stronger closure]\label{def:strictly-stronger}
For closure operators $c,d$ on $L$, write $c \prec d$ (``$d$ is strictly stronger than $c$'') if
$c\le d$ and there exists an $x\in L$ such that $d(x)\neq c(x)$.
\end{definition}

\begin{definition}[Closure ladder]\label{def:closure-ladder}
A \emph{closure ladder} is a sequence of closure operators $(c_n)_{n\in\mathbb N}$ on $L$ such that
$c_n \prec c_{n+1}$ for all $n$. The fixed-point sets form a nested chain
\[
\mathrm{Fix}(c_0)\supseteq \mathrm{Fix}(c_1)\supseteq \mathrm{Fix}(c_2)\supseteq \cdots
\]
by Lemma~\ref{lem:closed-antitone}.
\end{definition}

\begin{remark}
The strictness relation $\prec$ is irreflexive and transitive on closure operators (by the pointwise order),
and Lemma~\ref{lem:closed-antitone} records the antitone behavior of fixed points.
On a fixed finite lattice, any $\prec$-chain is finite; the open-endedness modeled here comes from changing
the theory/closure package or enlarging domains, not from a single finite lattice.
Viewed abstractly, a closure operator acts like a completion rule. Iterating a fixed completion rule stabilizes.
In a poset, closure operators correspond to reflections onto their fixed points (see Remark~\ref{rem:lean-closure}).
\end{remark}

\begin{remark}[Lean formalization note]\label{rem:lean-closure}
The order-theoretic closure operator used here is formalized in the repository using mathlib's
\texttt{ClosureOperator}. The basic closure and ladder API lives in
\path{formal/ClosureLadder/Basic.lean}. In particular, the iterate-stabilization statement
Lemma~\ref{lem:closure-iterate-stabilizes} corresponds to Lean lemmas
\texttt{closure\allowbreak *iterate\allowbreak *succ} and
\texttt{closure\allowbreak *iterate\allowbreak *ge\allowbreak *one}, and ladder monotonicity
corresponds to \texttt{ladder\allowbreak *mono}.
\end{remark}

% Scope IDs: see docs/spec/theorem-inventory.md and paper-contract.md
\section{Idempotent endomaps and induced closures}
\subsection{Idempotent endomaps}
\label{sec:idempotent-endo}

The common mathematical nucleus behind both order-theoretic closures and the
dynamics-induced operators used later is \emph{idempotence}.  In this section we
package that nucleus in a way that avoids conflating two different notions:
\emph{order-closures} (which require an ambient poset) versus \emph{idempotent
endomaps} (which require no order at all).

\begin{definition}[Idempotent endomap]
Let $A$ be a set (or type). An \emph{idempotent endomap} on $A$ is a function
$e : A \to A$ such that $e \circ e = e$, i.e.\ $e(e(a)) = e(a)$ for all $a\in A$.
\end{definition}

\begin{definition}[Fixed points]
For an endomap $e : A \to A$ define its fixed-point set
\[
\mathrm{Fix}(e) := \{a\in A : e(a)=a\}.
\]
We interpret $\mathrm{Fix}(e)$ as the collection of ``objects'' completed by $e$.
\end{definition}

\begin{lemma}[Idempotents split]
Let $e : A \to A$ be an idempotent endomap. Let $i:\mathrm{Fix}(e)\hookrightarrow A$
denote inclusion, and define $r:A\to \mathrm{Fix}(e)$ by $r(a):=e(a)$ (viewed as an
element of $\mathrm{Fix}(e)$).
Then $r\circ i = \mathrm{id}_{\mathrm{Fix}(e)}$ and $i\circ r = e$.
In particular, $e$ restricts to the identity on $\mathrm{Fix}(e)$ and
$\mathrm{Fix}(e)$ canonically identifies with the image of $e$.
\end{lemma}

\begin{proof}
For $a\in A$, idempotence gives $e(e(a))=e(a)$, so $e(a)\in \mathrm{Fix}(e)$ and
$r$ is well-defined. For $x\in \mathrm{Fix}(e)$ we have $r(i(x))=e(x)=x$, hence
$r\circ i=\mathrm{id}$. Finally, for $a\in A$, $(i\circ r)(a)=i(e(a))=e(a)$ by
definition of $i$, so $i\circ r=e$.
\end{proof}

\begin{remark}[Order-closures as a special case]
If $L$ is a poset and $c:L\to L$ is an order-closure (extensive, monotone,
idempotent), then forgetting the order yields an idempotent endomap on the
underlying set. The converse is false: an idempotent endomap need not be monotone
or extensive with respect to any nontrivial order.  In particular, the
dynamics-induced operators $E_{\tau,f}$ (defined later on probability simplices)
are treated as \emph{(approximate) idempotent endomaps} with fixed points and
idempotence defects measured in a chosen metric, rather than as order-closures.
We will sometimes call an idempotent endomap a \emph{semantic completion map} for a theory:
its fixed points are the internally complete descriptions with respect to that completion.
\end{remark}

\begin{remark}[Formalization note]
The constructions in this subsection are formalized in Lean:
\begin{itemize}
  \item \texttt{formal/ClosureLadder/}\\
  \texttt{IdempotentEndo.lean}
  \item Definitions: \texttt{ClosureLadder.IdempotentEndo}\\
  \texttt{ClosureLadder.IdempotentEndo.FixedPoints}.
  \item Bridge lemma:\\
  \texttt{ClosureLadder.\allowbreak ClosureOperator.}\\
  \texttt{toIdempotentEndo}.
\end{itemize}
\end{remark}

\subsection{Dynamics-induced empirical endomaps}
\label{sec:empirical-closure}

Throughout this section we assume \textbf{A\_FIN} and fix a finite state space $Z$.
Let $P$ be a (row-stochastic) Markov kernel on $Z$, and let $\Delta(Z)$ denote the simplex
of probability distributions on $Z$ (row vectors).
Fix a \emph{lens} $f:Z\to X$ (assumption \textbf{A\_LENS}), with fibers (blocks)
$B_x:=f^{-1}(x)$.

\begin{definition}[Coarse map and canonical lift; \textbf{D-IC-01}]
\label{def:D-IC-01}
Let $Q_f:\Delta(Z)\to \Delta(X)$ be the pushforward (coarse-graining) map
\[
(Q_f\mu)(x) \;:=\; \sum_{z\in B_x}\mu(z).
\]
Fix a \emph{canonical lift} $U_f:\Delta(X)\to \Delta(Z)$ specified by a choice of
\emph{prototype} distributions $(u_x)_{x\in X}$ with $\mathrm{supp}(u_x)\subseteq B_x$ and
$\sum_{z\in B_x}u_x(z)=1$, by setting
\[
(U_f\nu)(z) \;:=\; \sum_{x\in X}\nu(x)\,u_x(z).
\]
Equivalently, $U_f\delta_x=u_x$ and $Q_f U_f = \mathrm{id}_{\Delta(X)}$.

\paragraph{Existence Requires Choosing a Scale.}
For an integer time scale $\tau\ge 1$, define the \emph{induced empirical endomap}
\[
E_{\tau,f}:\Delta(Z)\to \Delta(Z),
\qquad
E_{\tau,f}(\mu) \;:=\; U_f\bigl(Q_f(\mu P^\tau)\bigr).
\]
\end{definition}

The pair $(f,E_{\tau,f})$ may be viewed as an \emph{empirical theory at scale $\tau$}:
$f$ specifies the expressive content (definability), while $E_{\tau,f}$ specifies a
semantics/completion rule whose fixed points are the objects recognized by the theory.

\begin{remark}[Not an order-closure]
The map $E_{\tau,f}$ is an affine endomap of $\Delta(Z)$ (indeed a Markov kernel on $Z$),
and it is always idempotent as an \emph{endomap} only in special regimes.
We treat $E_{\tau,f}$ within the ``fixed points of (approx) idempotents'' framework
(\S\ref{sec:idempotent-endo}), not as an order-closure operator.
See the repo spec note \texttt{docs/spec/empirical-vs-order-closure.md}.
\end{remark}

\begin{definition}[TV idempotence defect; \textbf{D-IC-02}]
\label{def:D-IC-02}
Let $\|\cdot\|_{\mathrm{TV}}$ denote total variation on $\Delta(Z)$, i.e.
$\|\mu-\nu\|_{\mathrm{TV}}=\tfrac12\|\mu-\nu\|_1$.
The \emph{idempotence defect} of $E_{\tau,f}$ is
\[
\delta_{\tau,f}
\;:=\;
\sup_{\mu\in \Delta(Z)} \bigl\|E_{\tau,f}(E_{\tau,f}(\mu)) - E_{\tau,f}(\mu)\bigr\|_{\mathrm{TV}}.
\]
\emph{Guardrail.} A small idempotence defect $\delta(E)$ is a \emph{saturation} diagnostic (reapplying $E$ changes $E(\mu)$ only slightly), but it does not by itself certify \emph{nontrivial} emergence or multiplicity (e.g., a constant map has $\delta(E)=0$ but only one fixed point). Any “multiple objects” conclusion therefore requires a separate nontriviality condition (for example: at least two distinct stable labels/prototypes at the chosen tolerance).
Since $\mu\mapsto \| \mu M\|_1$ is convex for fixed $M$ and $\Delta(Z)$ is the convex hull
of the Dirac masses $(\delta_z)_{z\in Z}$, the supremum is attained on extreme points:
\[
\delta_{\tau,f}
\;=\;
\max_{z\in Z} \bigl\|\delta_z(E_{\tau,f}^2 - E_{\tau,f})\bigr\|_{\mathrm{TV}}
\;=\;
\frac12\max_{z\in Z}\sum_{z'\in Z}\bigl|(E_{\tau,f}^2 - E_{\tau,f})(z,z')\bigr|.
\]
\end{definition}

\begin{definition}[Prototype stability at scale $\tau$]
For each $x\in X$ define the (scale-$\tau$) stability of the prototype $u_x$ by
\[
\begin{aligned}
s_{\tau,f}(x) &:= \bigl\|E_{\tau,f}(u_x)-u_x\bigr\|_{\mathrm{TV}}\\
&= \bigl\|U_f(Q_f(u_xP^\tau)) - U_f(\delta_x)\bigr\|_{\mathrm{TV}}.
\end{aligned}
\]
We say $u_x$ is \emph{$\varepsilon$-stable} if $s_{\tau,f}(x)\le \varepsilon$.
\end{definition}

\begin{theorem}[Approximate idempotence from retention]
\label{thm:T-IC-02}
\textbf{T-IC-02.}\par
Assume \textbf{A\_FIN}+\textbf{A\_LENS} and fix $E_{\tau,f}$ as in Definition~\ref{def:D-IC-01}.
Define a \emph{retention error}
\[
\varepsilon_{\tau,f}
\;:=\;
\max_{x\in X}\bigl\|Q_f(u_xP^\tau)-\delta_x\bigr\|_{\mathrm{TV}}.
\]
Then:
\begin{enumerate}
\item (Prototype stability bound) For all $x\in X$, $s_{\tau,f}(x)\le \varepsilon_{\tau,f}$.
\item (Idempotence defect bound) $\delta_{\tau,f}\le \varepsilon_{\tau,f}$.
In particular, if $\varepsilon_{\tau,f}=0$ then $E_{\tau,f}$ is exactly idempotent.
\end{enumerate}
\end{theorem}

\begin{proof}
Since $U_f$ is a stochastic map, it is a contraction in total variation.
Thus for each $x$,
\[
\begin{aligned}
\|E_{\tau,f}(u_x)-u_x\|_{\mathrm{TV}}
&=
\|U_f(Q_f(u_xP^\tau)) - U_f(\delta_x)\|_{\mathrm{TV}}\\
&\le
\|Q_f(u_xP^\tau)-\delta_x\|_{\mathrm{TV}}
\le \varepsilon_{\tau,f}.
\end{aligned}
\]
This proves (1).

For (2), write any $\mu\in\Delta(Z)$ as $\mu_1:=E_{\tau,f}(\mu)=\sum_{x\in X}w_x u_x$
where $w:=Q_f(\mu P^\tau)\in \Delta(X)$ (so $w_x\ge 0$ and $\sum_x w_x=1$).
Then by linearity,
\[
E_{\tau,f}(\mu_1) - \mu_1
=
\sum_{x\in X} w_x\bigl(E_{\tau,f}(u_x)-u_x\bigr).
\]
Taking TV norms and using the bound from (1),
\[
\begin{aligned}
\|E_{\tau,f}(\mu_1) - \mu_1\|_{\mathrm{TV}}
&\le
\sum_{x\in X}w_x \|E_{\tau,f}(u_x)-u_x\|_{\mathrm{TV}}\\
&\le
\sum_{x\in X}w_x \varepsilon_{\tau,f}
=
\varepsilon_{\tau,f}.
\end{aligned}
\]
Now take the supremum over $\mu$ to obtain $\delta_{\tau,f}\le \varepsilon_{\tau,f}$.
\end{proof}

Refinement is not monotone-good in general: it can reveal additional stable objects when aligned with metastable structure, but it can also destroy stability when it misaligns with mixing/leakage scales.

\begin{theorem}[Refinement can help or hurt]
\label{thm:T-IC-01}
\textbf{T-IC-01.}\par
Assume \textbf{A\_FIN}+\textbf{A\_LENS}.
There exist finite Markov kernels $P$ and times $\tau$ together with two lenses
\[
f_{\mathrm{fine}}:Z\to X_{\mathrm{fine}},\qquad f_{\mathrm{coarse}}:Z\to X_{\mathrm{coarse}},
\]
such that $f_{\mathrm{coarse}}$ is a coarsening of $f_{\mathrm{fine}}$ (the fine partition refines the coarse),
and for the canonical lift $U_f$ defining $E_{\tau,f}$ (Definition~\ref{def:D-IC-01}),
the number of $\varepsilon$-stable prototypes can either increase or decrease under refinement
for a fixed stability threshold $\varepsilon>0$.

More precisely, both behaviors occur:
\begin{enumerate}
\item (\emph{Refinement reveals objects}) in one example, refinement increases the number of
$\varepsilon$-stable prototypes.
\item (\emph{Refinement destroys objects}) in another example, the stable count decreases.
\end{enumerate}
\end{theorem}

\begin{proof}[Proof idea by explicit two-block constructions]
The mechanism is controlled by the retention errors $\varepsilon_{\tau,f}$ of
Theorem~\ref{thm:T-IC-02}.

Consider a state space split into two blocks $B_0,B_1$ and a lift $U_f$ whose prototypes
$u_0,u_1$ are supported on the blocks. If the dynamics jump from $B_0$ to $B_1$ with
probability $p$ each step (and otherwise remain in the current block), then starting from $u_0$,
the probability of exiting $B_0$ within $\tau$ steps is $1-(1-p)^\tau$.
Hence $\|Q_f(u_0P^\tau)-\delta_0\|_{\mathrm{TV}}$ is on the order of $1-(1-p)^\tau$,
and similarly for $u_1$.

Now take a coarse lens that merges two metastable regions into one block, and a refinement that
splits them. If $p$ is very small (relative to $\tau$), the refined prototypes have small
retention error and are stable, so refinement increases the stable count.
If $p$ is not small (relative to $\tau$), the refined prototypes rapidly mix across the split
boundary, their retention errors exceed the stability threshold, and refinement decreases the
stable count.  \qedhere
\end{proof}

\begin{remark}[Qualified ``theory birth'']
The induced endomap $E_{\tau,f}$ packages a distribution by evolving for $\tau$ steps and then
forgetting micro-detail below the lens $f$, re-instantiating it via the chosen prototypes.
Theorem~\ref{thm:T-IC-02} shows that a small retention error implies small idempotence defect.
In this regime, the (approximately) stable prototypes $u_x$ behave as emergent ``objects'' at
scale $\tau$ in the theory $f$.
Refinement changes the lens and can therefore reveal or destroy such objects
(Theorem~\ref{thm:T-IC-01}).
\end{remark}

% Scope IDs: see docs/spec/theorem-inventory.md and paper-contract.md
\section{AUT + REV + ACC regime and graph 1-forms}
\label{sec:acc}

This section isolates the representation-independent content behind the admissibility regime
\textbf{A\_AUT}+\textbf{A\_REV}+\textbf{A\_ACC} that we use throughout: the obstruction to
``no-drive'' behavior is a cohomological (cycle) obstruction on the bidirected support graph.
Our presentation is discrete-time and purely graph-theoretic; any continuous-time (CTMC) reading
is optional and out of scope for the main text.

\subsection{Bidirected support and the log-ratio 1-form}

Assume \textbf{A\_FIN} and let $P$ be a time-homogeneous Markov kernel on a finite state space $Z$.
Write $P_{ij}:=P(i,j)$.
Under \textbf{A\_REV} (microreversibility on support), we have
\[
P_{ij}>0 \implies P_{ji}>0,
\]
so every allowed transition comes with a reverse transition.

Define the (undirected) \emph{support graph} $G=(Z,E)$ by declaring $\{i,j\}\in E$ iff $P_{ij}>0$
(equivalently $P_{ji}>0$ by \textbf{A\_REV}). Let $\vec E$ be the corresponding set of
\emph{oriented} edges: for each $\{i,j\}\in E$ we include both $(i,j)$ and $(j,i)$.

\begin{definition}[Edge log-ratio 1-form]
\label{def:edge-one-form}
Assume \textbf{A\_FIN}+\textbf{A\_REV}. The \emph{antisymmetric edge 1-form} associated to $P$ is
the function $a:\vec E\to \mathbb R$ defined by
\[
a_{ij} := a(i,j) := \log\frac{P_{ij}}{P_{ji}} \qquad \text{for } (i,j)\in \vec E.
\]
Then $a_{ji}=-a_{ij}$.
\end{definition}

The cycle/affinity viewpoint for Markov networks and its cohomological structure is standard; see
e.g.\ \cite{Schnakenberg1976,Polettini2012,AltanerEtAl2012,DalCengioLecomtePolettini2023}.

\begin{remark}[Why \textbf{A\_REV} matters]
If \textbf{A\_REV} fails, then ratios such as $P_{ij}/P_{ji}$ may be undefined, and many
time-reversal quantities become infinite rather than merely nonzero. In this paper we treat
\textbf{A\_REV} as the basic finiteness/admissibility condition for edge log-ratios.
\end{remark}

\subsection{Cycle integrals, exactness, and the null regime}

A (simple) cycle in $G$ is a sequence $(v_0,v_1,\dots,v_m=v_0)$.
We require $\{v_k,v_{k+1}\}\in E$ for all $k$.

\begin{definition}[Cycle integral / affinity along a cycle]
\label{def:cycle-integral}
Given an antisymmetric edge 1-form $a$ and an oriented cycle $\gamma=(v_0,\dots,v_m=v_0)$, define
\[
\mathcal A(\gamma) := \sum_{k=0}^{m-1} a_{v_k v_{k+1}}.
\]
Reversing the orientation negates $\mathcal A(\gamma)$.
\end{definition}

\begin{definition}[Exact 1-forms]
\label{def:exact-one-form}
An antisymmetric 1-form $a$ on $\vec E$ is \emph{exact} if there exists a potential
$\Phi:Z\to\mathbb R$ such that for every $(i,j)\in \vec E$,
\[
a_{ij} = \Phi(j)-\Phi(i).
\]
We write $a=d\Phi$ in this case.
\end{definition}

\begin{lemma}[Exact forms have zero cycle integrals]
\label{lem:exact-zero-cycles}
\mbox{}\par
Assume \textbf{A\_FIN}+\textbf{A\_REV}.
If $a=d\Phi$ is exact, then $\mathcal A(\gamma)=0$ for every cycle $\gamma$ in $G$.
\end{lemma}

\begin{proof}
Along a cycle $\gamma=(v_0,\dots,v_m=v_0)$ we have
\[
\mathcal A(\gamma)=\sum_{k=0}^{m-1}\big(\Phi(v_{k+1})-\Phi(v_k)\big)=\Phi(v_m)-\Phi(v_0)=0.
\]
\end{proof}

\paragraph{Force Lives on Loops.}
\begin{theorem}[Cycle criterion for exactness]
\label{thm:cycle-criterion-exact}
Assume \textbf{A\_FIN}+\textbf{A\_REV} and let $a$ be an antisymmetric 1-form on the oriented edge
set $\vec E$ of $G$. The following are equivalent:
\begin{enumerate}
\item $a$ is exact: $a=d\Phi$ for some $\Phi:Z\to\mathbb R$.
\item $\mathcal A(\gamma)=0$ for every cycle $\gamma$ in $G$.
\item $\mathcal A(\gamma)=0$ for every cycle $\gamma$ in some fixed cycle basis of $G$
(e.g.\ a fundamental cycle basis from a spanning tree).
\end{enumerate}
\end{theorem}

\begin{proof}[Proof sketch]
(1)$\Rightarrow$(2) is Lemma~\ref{lem:exact-zero-cycles}. (2)$\Rightarrow$(3) is immediate.

For (3)$\Rightarrow$(1), fix a spanning tree $T$ of each connected component of $G$.
Choose a root $r$ in each component and define $\Phi(r)=0$; for any vertex $v$ in that component,
let $\Phi(v)$ be the sum of $a$ along the unique tree path from $r$ to $v$ (with signs determined
by orientations).
This makes $a_{ij}=\Phi(j)-\Phi(i)$ hold for every tree edge $(i,j)$.
For any chord edge $e=(u,v)\notin T$, the fundamental cycle $\gamma_e$ is the tree path from $u$
to $v$ together with $e$; the hypothesis $\mathcal A(\gamma_e)=0$ forces $a_{uv}=\Phi(v)-\Phi(u)$
on the chord as well.
Thus $a=d\Phi$ on all of $\vec E$.
\end{proof}

\begin{corollary}[Null regime as ``no-drive'' / detailed-balance-like]
\label{cor:null-regime}
Assume \textbf{A\_FIN}+\textbf{A\_REV}+\textbf{A\_NULL} and let $a$ be the log-ratio 1-form on the bidirected support graph.
Then $a$ is exact.
Equivalently, there exists $\Phi:Z\to\mathbb R$ such that
\[
\log\frac{P_{ij}}{P_{ji}} = \Phi(j)-\Phi(i) \quad \text{for all } (i,j)\in \vec E.
\]
In particular, defining $\pi(i)\propto e^{\Phi(i)}$ yields the detailed balance relations
$\pi(i)P_{ij}=\pi(j)P_{ji}$ on support.
\end{corollary}

\subsection{Accounting as coordinates on cycle space}

\paragraph{Drive Is Coordinate-Free.}
Let $\beta_1(G)$ be the cycle rank of the undirected support graph.
Choosing a cycle basis $(\gamma_1,\dots,\gamma_{\beta_1})$ defines a coordinate vector of
\emph{affinities}
\[
A_r := \mathcal A(\gamma_r), \qquad r=1,\dots,\beta_1(G).
\]
Changing the cycle basis changes coordinates but not the underlying obstruction: $A=0$ in one
basis iff $A=0$ in every basis (Theorem~\ref{thm:cycle-criterion-exact}).

We regard \textbf{A\_ACC} as the decision to work with this decomposition viewpoint:
the antisymmetric log-ratio 1-form $a$ splits into an exact (gauge) part plus a cycle/affinity
part, and the affinity part vanishes exactly in the null regime \textbf{A\_NULL}.

\begin{remark}[Implementation note for Appendix C]
The repository includes a small graph-theoretic routine that computes cycle affinities and tests
exactness on finite Markov kernels (reversible vs biased 3-cycle examples). We relegate all such
numerical evidence to Appendix~C per the paper contract.
\end{remark}

% Scope IDs: see docs/spec/theorem-inventory.md and paper-contract.md
\section{Arrow-of-time as path reversal asymmetry; data processing; protocol trap}\label{sec:aot}

Throughout this section we assume \textbf{A\_FIN}. Let $Z$ be a finite state space and let
$P:Z\times Z\to [0,1]$ be a time-homogeneous Markov kernel.

\begin{definition}[Forward path law and reversal (D-AOT-01)]
Fix a horizon $T\in \mathbb N$ with $T\ge 1$ and an initial distribution $\rho\in\Delta(Z)$.
The (forward) length-$T$ path law $\mathbb P_{\rho,T}$ is the probability measure on $Z^{T+1}$ defined by
\[
\mathbb P_{\rho,T}(z_0,\dots,z_T)
:= \rho(z_0)\prod_{t=0}^{T-1} P(z_t,z_{t+1}).
\]
Let $\mathcal R:Z^{T+1}\to Z^{T+1}$ be the reversal map $\mathcal R(z_0,\dots,z_T)=(z_T,\dots,z_0)$, and write $\mathcal{R}_{*}\mu$ for pushforward of a measure $\mu$.
Define the (finite-horizon) \emph{path reversal asymmetry}
\[
\Sigma_T(\rho)\;:=\; D_{\mathrm{KL}}\!\left(\mathbb P_{\rho,T}\,\middle\|\,\mathcal{R}_{*}\mathbb P_{\rho,T}\right)\in[0,\infty].
\]
We use the standard convention $D_{\mathrm{KL}}(p\|q)=+\infty$ if $p$ is not absolutely continuous with respect to $q$.
\end{definition}

Classical path-space fluctuation formulations of time-reversal asymmetry appear in \cite{LebowitzSpohn1999}.

\begin{remark}[Stationarity is not required (L-AOT-STAT-01)]
The quantity $\Sigma_T(\rho)$ is defined for any initial distribution $\rho$; it can be strictly positive even in a detailed-balance chain when $\rho$ is not stationary (a boundary-term effect).
Stationarity is only needed when we want a \emph{steady-state} interpretation (entropy production rate), not for the definition of $\Sigma_T$ or for the data-processing inequality below.
\end{remark}

\begin{definition}[Steady-state entropy production (D-AOT-EP-01)]
Assume additionally \textbf{A\_STAT} and let $\pi$ be a stationary distribution for $P$.
The (finite-horizon) steady-state entropy production is $\mathrm{EP}_T:=\Sigma_T(\pi)$.
When a per-step rate is desired, we write $\mathrm{ep}:=\tfrac1T\Sigma_T(\pi)$ (and, under standard ergodicity conditions, one can pass to the limit $T\to\infty$).
\end{definition}

\subsection{Data processing: coarse-graining cannot create asymmetry}

\paragraph{No Fake Arrows.}
Now assume \textbf{A\_LENS} and fix a coarse-graining (lens) $f:Z\to X$.
Let $f^{T+1}:Z^{T+1}\to X^{T+1}$ be the coordinatewise map
\[
f^{T+1}(z_0,\dots,z_T)=(f(z_0),\dots,f(z_T)).
\]
Define the observed path law $\mathbb Q_{\rho,T}:=(f^{T+1})_*\mathbb P_{\rho,T}$.
Note that the observed process on $X$ need not be Markov; the DPI statement below is purely about path measures.

\begin{theorem}[Data processing for path reversal asymmetry (T-AOT-01)]\label{thm:dpi_path}
Under \textbf{A\_FIN}+\textbf{A\_LENS}, for every $T\ge 1$ and every initial distribution $\rho$,
\[
D_{\mathrm{KL}}\!\left(\mathbb Q_{\rho,T}\,\middle\|\,\mathcal{R}_{*}\mathbb Q_{\rho,T}\right)
\;\le\;
D_{\mathrm{KL}}\!\left(\mathbb P_{\rho,T}\,\middle\|\,\mathcal{R}_{*}\mathbb P_{\rho,T}\right)
=\Sigma_T(\rho).
\]
\end{theorem}

\begin{proof}[Proof sketch]
The reversal map $\mathcal R$ commutes with coordinatewise coarse-graining:
\[
\mathcal R\circ f^{T+1}=f^{T+1}\circ \mathcal R.
\]
Hence $\mathcal{R}_{*}\mathbb Q_{\rho,T}=(f^{T+1})_*(\mathcal{R}_{*}\mathbb P_{\rho,T})$.
The claim is then the usual contraction of KL under measurable maps: for any measurable $g$, $D_{\mathrm{KL}}(g_*p\|g_*q)\le D_{\mathrm{KL}}(p\|q)$.
\end{proof}

Importantly, this monotonicity rules out \emph{false positives}: coarse observation can hide irreversibility but cannot create it.

\begin{corollary}[No false positives]
Assume \textbf{A\_FIN}+\textbf{A\_LENS}.
If the lifted system has zero asymmetry at horizon $T$, then every coarse observation has zero asymmetry at the same horizon.
There are no false positives from forgetting variables.
Coarse-graining can hide irreversibility (false negatives), but it cannot create it.
\end{corollary}

Related coarse-graining and hidden-drive discussions in stochastic thermodynamics include
\cite{Puglisi2010,DeguntherVanderMeerSeifert2023,BauerSeifertVanderMeer2025,Merhav2010}.

Equivalently: passing to a coarser theory (via a lens pushforward) cannot create time-reversal asymmetry.

\subsection{Protocol trap: apparent stroboscopic arrows and the ``clock audit''}

\paragraph{P3 Loves P6 Law.}
We now pin down a class of autonomous protocol models and the sense in which ``P3 requires P6\_drive''.
Let $X$ be a microstate space and let $\Phi=\{0,1,\dots,m-1\}$ be a finite phase space.
A \emph{protocol family} is a collection of kernels $\{K_\phi\}_{\phi\in\Phi}$ on $X$ together with a phase kernel $\mathsf S$ on $\Phi$.

\begin{definition}[Autonomous $m$-phase protocol]
\textbf{D-PROT-02.}\par
Fix $\alpha\in[0,1]$ and define an autonomous Markov chain on the product space $Z:=X\times \Phi$ by
\[
P\big((x,\phi),(x',\phi')\big)
:=\alpha\,\mathbf{1}_{x'=x}\,\mathsf{S}\,(\phi,\phi')
+(1-\alpha)\,\mathbf{1}_{\phi'=\phi}\,K_\phi(x,x').
\]
We call this the \emph{lifted} (phase-included) protocol model.
\end{definition}

\begin{theorem}[Protocol trap / hidden-drive principle (T-AOT-02)]\label{thm:protocol_trap}
Assume \textbf{A\_FIN}+\textbf{A\_AUT}.
Suppose:
(i) the phase chain $\mathsf S$ is reversible with stationary distribution $s$;
(ii) each $K_\phi$ is reversible with respect to a \emph{common} stationary distribution $\pi$ on $X$.
Then the product measure $\mu:=\pi\times s$ is stationary for the lifted chain on $Z=X\times\Phi$ and the lifted chain is reversible (hence $\Sigma_T(\mu)=0$ for all $T$).
Consequently, by Theorem~\ref{thm:dpi_path} applied to the projection $X\times\Phi\to X$, the observed process on $X$ has zero path reversal asymmetry at every finite horizon.

In particular, any \emph{apparent} arrow-of-time obtained by analysing a stroboscopic composition of kernels on $X$ must be traced to either:
(a) an external schedule (violating \textbf{A\_AUT}), or
(b) a driven/bias phase dynamics (nonzero affinity, cf.\ Section~\ref{sec:acc}).
\end{theorem}
\begin{proof}[Proof sketch]
Let $P_{\mathrm{phase}}$ and $P_{\mathrm{state}}$ be the phase-update-only and state-update-only kernels defined by
\[
\begin{aligned}
P_{\mathrm{phase}}((x,\phi),(x',\phi'))&=\mathbf{1}_{x'=x}\mathsf{S}(\phi,\phi'),\\
P_{\mathrm{state}}((x,\phi),(x',\phi'))&=\mathbf{1}_{\phi'=\phi}K_\phi(x,x').
\end{aligned}
\]
By reversibility of $\mathsf S$ with respect to $s$ and of each $K_\phi$ with respect to the common $\pi$,
both $P_{\mathrm{phase}}$ and $P_{\mathrm{state}}$ are reversible with respect to $\mu=\pi\times s$.
The lifted kernel is their convex combination
$P=\alpha P_{\mathrm{phase}}+(1-\alpha)P_{\mathrm{state}}$, hence is also reversible with respect to $\mu$.
Therefore $\Sigma_T(\mu)=0$ for all $T$, and the DPI argument yields zero asymmetry on $X$.
\end{proof}

\begin{corollary}[``P3 needs P6\_drive'' under autonomy (C-ACC-01)]
Under \textbf{A\_FIN}+\textbf{A\_REV}+\textbf{A\_AUT}+\textbf{A\_ACC}+\textbf{A\_NULL} (no affinities anywhere in the lifted state space),
protocol holonomy alone does not yield sustained arrow-of-time in the autonomous accounted model.
For any stationary distribution of the lifted autonomous chain, nonzero steady-state entropy production requires a nontrivial affinity component (P6\_drive)
in the lifted dynamics (e.g.\ a biased phase cycle) or an external schedule.
\end{corollary}

\begin{remark}[What the ``protocol trap'' actually warns about]
The theorem above is about \emph{true} path reversal asymmetry.
In practice, one can still obtain a spurious ``arrow'' by fitting a one-step Markov model to a non-Markov observed process (for example, by replacing the hidden-phase model by a stroboscopic kernel on $X$).
The cure is the clock audit: include the phase variable in the state, or work directly with path-space quantities that do not assume Markovianity at the observed level.
\end{remark}

% Scope IDs: see docs/spec/theorem-inventory.md and paper-contract.md
\section{Generic extension and the finite forcing lemma}
\label{sec:forcing}
\subsection{Theories as partitions and definability}

Assume \textbf{A\_FIN} and fix a finite microstate set $Z$.
A \emph{theory} (or \emph{lens}) is a finite-valued map $f:Z\to X$ (assumption \textbf{A\_LENS}), inducing the partition
\[
\Pi_f := \{B_x\}_{x\in X}, \qquad B_x := f^{-1}(x).
\]
We interpret $\Pi_f$ as the collection of distinctions expressible in the current theory.

\begin{definition}[Predicates and definability]
A (Boolean) \emph{predicate} is a function $h:Z\to\{0,1\}$.
We say that $h$ is \emph{definable from the theory $f$} (equivalently: $\Pi_f$-measurable) if it is constant on each block:
\[
\forall x\in X,\ \forall z,z'\in B_x,\quad h(z)=h(z').
\]
Equivalently, there exists $\tilde h:X\to\{0,1\}$ with $h=\tilde h\circ f$.
\end{definition}

Finite-model-theoretic context for definability in finite structures and counting viewpoints
appears in, e.g., \cite{MacphersonSteinhorn2011,Hyttinen2015}.

\begin{remark}
In a fixed finite $Z$ there is a hard cap on how many strict refinements are possible; open-endedness is
modeled by repeated extensions across growing problems/domains (each step finite).
The ``finite forcing lemma'' is a combinatorial proxy and does not claim set-theoretic independence.
\end{remark}

Adjoining a new predicate $h$ refines the theory by considering the joint lens
\[
(f,h): Z \to X\times\{0,1\}.
\]
This extension is \emph{trivial} precisely when $h$ was already definable from $f$.

\subsection{Counting lemma: definable predicates are rare}

Let $N:=|Z|$ and $K:=|X|$.

\begin{lemma}[Counting definable predicates]\label{lem:count-definable}
Assume \textbf{A\_FIN}+\textbf{A\_LENS}.
The number of predicates definable from $f:Z\to X$ is exactly $2^K$.
Consequently, if $h$ is sampled uniformly from $\{0,1\}^Z$, then
\[
\mathbb P(h \text{ is definable from } f)=\frac{2^K}{2^N}=2^{-(N-K)}.
\]
\end{lemma}

\begin{proof}
A definable predicate is determined by choosing one bit per block $B_x$ (the common value on that block), hence $2^K$ choices.
There are $2^N$ total predicates on $Z$, so the stated probability follows.
\end{proof}

\subsection{Finite forcing: generic extensions are non-definable}

The preceding lemma is the finite analogue of the forcing slogan: \emph{generic extensions add non-definable predicates}.
To emphasize the “genericity” aspect, we also record a simple sufficient condition for a predicate to be \emph{strongly novel}.

\paragraph{Almost Nothing Is Definable.}
\begin{theorem}[Finite forcing lemma / physical forcing lemma]\label{thm:finite-forcing}
Assume \textbf{A\_FIN}+\textbf{A\_LENS}.
Let $f:Z\to X$ be a theory with blocks $\Pi_f=\{B_x\}_{x\in X}$, and sample $h:Z\to\{0,1\}$ uniformly from $\{0,1\}^Z$ (i.i.d.\ fair bits across microstates).
Then \textbf{non-definability is generic:}
\[
\mathbb P(h \text{ is not definable from } f)=1-2^{-(N-K)}.
\]
\end{theorem}

\begin{proof}
This is Lemma~\ref{lem:count-definable}.
\end{proof}

\paragraph{The ``Nothing Stays Constant'' Lemma.}
\begin{lemma}
Assume \textbf{A\_FIN}+\textbf{A\_LENS} and let $f:Z\to X$ have blocks $\Pi_f=\{B_x\}_{x\in X}$.
For each block $B_x$,
\[
\mathbb P(h \text{ is constant on } B_x)=2^{1-|B_x|},
\]
and by a union bound,
\[
\mathbb P(\exists x\in X\text{ such that }h\text{ is constant on }B_x)\ \le\ \sum_{x\in X} 2^{1-|B_x|}.
\]
In particular, if every block has size at least $m\ge 2$, then
\[
\mathbb P(h \text{ splits every block}) \ \ge\ 1 - K\,2^{1-m}.
\]
\end{lemma}

\begin{proof}
On a block of size $|B_x|$, the only constant labelings are ``all zeros'' and ``all ones'', each with probability $2^{-|B_x|}$, hence $2^{1-|B_x|}$ total.
The union bound gives the stated inequality, and the final display uses $|B_x|\ge m$ for all $x$.
\end{proof}

\begin{corollary}[Generic extension is strict]\label{cor:strict-language-extension}
\mbox{}\par
Assume \textbf{A\_FIN}+\textbf{A\_LENS}.
With probability $1-2^{-(N-K)}$, the refined lens $(f,h):Z\to X\times\{0,1\}$ is a \emph{strict} refinement of $f$.
Equivalently, it distinguishes at least one pair of microstates previously indistinguishable under $f$.
\end{corollary}

\begin{proof}
If the refinement were not strict, then $h$ would be constant on each block of $f$, i.e.\ definable from $f$.
The complement event has probability $1-2^{-(N-K)}$ by Theorem~\ref{thm:finite-forcing}(1).
\end{proof}

\begin{remark}[Why this is the anti-saturation move]
Closure (or packaging) at a fixed theory saturates by idempotence: repeated closure steps do not create an infinite strict chain.
Corollary~\ref{cor:strict-language-extension} exhibits a mathematically clean \emph{extension step}: adjoining a generic predicate changes the theory itself in a way that is overwhelmingly not reconstructible from the old distinctions.
This is the finite “forcing-style” mechanism we use to justify strict ladder climbing via theory/completion change.
\end{remark}

% Scope IDs: see docs/spec/theorem-inventory.md and paper-contract.md
\section{Why the primitives are unavoidable}
\label{sec:meta-unavoidable}

\begin{definition}[Process soup (D-META-PROC-01)]\label{def:meta-proc}
A \emph{process soup} is a set $\mathcal P$ equipped with a partially defined associative composition
$\circ$ (a partial semigroup), with identities where relevant.
\end{definition}

\begin{definition}[Interface lens (D-META-LENS-01)]\label{def:meta-lens}
An \emph{interface lens} is a map $f:\mathcal P\to X$ to a set of observables. It induces a kernel
equivalence $x\sim_f y$ iff $f(x)=f(y)$. Limited access means $\exists\,x\neq y$ with $x\sim_f y$.
This abstracts the notion of a coarse-graining lens $f:Z\to X$ used elsewhere: here the lens acts on processes/happenings rather than on state spaces.
\end{definition}

\begin{definition}[Refinement family (D-META-REF-01)]\label{def:meta-ref}
A refinement family is a chain of equivalence relations $(\sim_j)_{j\ge 0}$ with
$x\sim_{j+1} y \Rightarrow x\sim_j y$; equivalently, lenses $f_j$ with
$f_j=g_j\circ f_{j+1}$ for some maps $g_j$.
As $j$ increases, observables become more discriminating (partition refinement), hence the induced
saturation operators on predicates become weaker (descending in the pointwise order).
\end{definition}

\begin{definition}[Bounded interface (D-META-BND-01)]\label{def:meta-bnd}
There exists $C_0\ge 1$ such that the quotient index satisfies
$|\mathcal P/\!\sim_j|\le C_0(j+1)$ for all $j$.
This is an explicit bounded-interface hypothesis on the lens hierarchy (a modeling assumption encoding bounded observable signature); it is not derived from refinement alone.
\end{definition}
For example, refinement alone permits exponential growth (e.g.\ $|X_j|=2^j$ with $X_j=\{0,1\}^j$), so the linear bound must be assumed or verified in an instantiation.

\begin{theorem}[Self-generated primitives]\label{thm:meta-prim}
Assume D-META-PROC-01, D-META-LENS-01, D-META-REF-01, and D-META-BND-01.
Then the six primitives P1--P6 appear canonically as closure mechanics.
Here \textbf{P5}, \textbf{P6}, \textbf{P4}, and \textbf{P2} are by construction from the lens/refinement data, while \textbf{P1} and \textbf{P3} isolate the only nontrivial compatibility requirements as named hypothesis slots (HL-META-1/2/3).
\begin{enumerate}
  \item \textbf{P5 (Packaging).} Each equivalence $\sim_j$ yields a quotient/packaging map
  $\Pi_j:\mathcal P\to \mathcal P/\!\sim_j$ and induces an idempotent saturation on predicates/events via
  $\mathrm{cl}_j(A):=\Pi_j^{-1}(\Pi_j(A))$ (cf.\ Sections~\ref{sec:closure} and \ref{sec:idempotent-endo}).
  \item \textbf{P6 (Accounting).} The refinement order on quotients induces a preorder and monotone ``audit'' quantities (e.g.\ $j\mapsto |X_j|$); this is an information/feasibility order only (not, by itself, a thermodynamic arrow-of-time).
  \item \textbf{P4 (Staging).} The refinement chain generates a depth index $j$, and bounded interfaces ensure $|X_j|\lesssim (j+1)$. Nontrivial staging corresponds to strict refinement at some scale (HL-META-2: $\sim_{j+1}\subsetneq \sim_j$ for some $j$).
  \item \textbf{P2 (Constraints).} Feasible macrostates are those representable as images under $\Pi_j$,
  yielding a constraint set carved by interface compatibility.
  \item \textbf{P1 (Operator rewrite).} Fix a micro-update $F:\mathcal P\to\mathcal P$. The induced macro-update
  $F^\sharp([p]) := [F(p)]$ on $\mathcal P/\!\sim_j$ is well-defined iff $p\sim_j q \Rightarrow F(p)\sim_j F(q)$ (HL-META-1).
  When (HL-META-1) fails, one must either (i) refine the lens so stability holds, (ii) modify/extend $F$ to respect $\sim_j$, or (iii) accept that no closed macro dynamics exists at depth $j$.
  \item \textbf{P3 (Holonomy).} Assuming there exist two admissible closure/evolution routes with the same input/output type (HL-META-3)
  (e.g.\ $\Pi_j\circ F$ versus $F^\sharp\circ \Pi_j$ when $F^\sharp$ is well-defined), their discrepancy defines a holonomy-like diagnostic.
  In normed operator instantiations we measure this by $\mathrm{RM}(j)$ (Appendix~\ref{app:tk-defects}, Eq.~\eqref{eq:tk-rm}); this is a discrepancy diagnostic, not itself a directionality certificate (see Section~\ref{sec:aot}).
 \end{enumerate}
\end{theorem}

\begin{proof}[Proof sketch]
P5 is the quotient/packaging map induced by the equivalence, which is idempotent by construction.
P6 follows from the refinement preorder and yields information/feasibility monotones only (not a thermodynamic arrow-of-time).
P4 is the stage index implicit in the refinement chain (nontrivial if refinements are strict, HL-META-2), which also connects to closure ladders (Section~\ref{sec:closure}).
P2 is the feasible image under $\Pi_j$, i.e.\ representability at the interface.
For P1, given a micro-update $F$, the expression $F^\sharp([p])=[F(p)]$ descends to a map on $\mathcal P/\!\sim_j$ exactly when (HL-META-1) holds; when it fails, closed macro dynamics at depth $j$ requires rewrite/extension, refinement, or must be abandoned.
For P3, whenever two admissible routes exist (HL-META-3), their noncommutativity produces a route-mismatch/holonomy defect, formalized by $\mathrm{RM}(j)$ in Appendix~\ref{app:tk-defects}; this defect is agnostic to directionality without additional structure (Section~\ref{sec:aot}).
\end{proof}

\section{Primitives P1--P6 as closure-changing operations}
\label{sec:primitives}
The primitives are presented below as a minimal operator vocabulary, but they are not axioms:
Theorem~\ref{thm:meta-prim} (Section~\ref{sec:meta-unavoidable}) shows they are structurally forced once
composability, limited access, and bounded interfaces are assumed. This framing remains compatible with
the corrected stance on P3: protocol/phase is part of the autonomous state (no external schedule), and
route dependence alone does not certify arrow-of-time.
Primitives P1--P6 form a vocabulary of operations that change (i) the admissible transition
structure (support graph and cycle space), (ii) the available idempotent endomaps (packaging and
fixed points), and (iii) the affinity data encoded by the graph 1-form. These are not ratchets by
themselves; directionality is certified by path reversal asymmetry and/or non-exactness of the
log-ratio 1-form under the \textbf{A\_AUT}+\textbf{A\_REV}+\textbf{A\_ACC} regime.

\subsection{Definitions of P1--P6}

\begin{definition}[P1: operator rewrite]
P1 is a rewrite of the substrate operator: replace a Markov kernel $P$ by a new kernel $P'$ (or a
finite family $\{P^{(m)}\}$), thereby changing the endomap $\mu\mapsto \mu P^\tau$ and the induced
empirical endomap $E_{\tau,f}$. This is a change in the operator itself, not an external schedule.
\end{definition}

\begin{definition}[P2: gating / constraints]
P2 restricts the support graph by deleting edges (setting selected $P_{ij}=0$) and renormalizing
rows, or equivalently restricting to a subgraph on which the kernel lives. Cycle-rank claims are
interpreted on the undirected support graph justified by \textbf{A\_REV}.
\end{definition}

\begin{definition}[P3: autonomous protocol holonomy]
P3 is modeled in the autonomous lifted form on $Z:=X\times\Phi$. The phase $\phi\in\Phi$ evolves
by an internal kernel $\mathsf S$, and conditioned on $\phi$ the microstate updates by $K_\phi$.
No external schedule is assumed in the main text; externally scheduled stroboscopic protocols are
non-autonomous and fall outside \textbf{A\_AUT}.
P3 (route mismatch/holonomy) is a \emph{protocol-geometry diagnostic}: by itself it does not certify directionality, and any arrow-of-time claim must be supported by an audit functional (e.g.\ affinity or path reversal asymmetry) in a model where the protocol state is included.
\end{definition}

\begin{definition}[P4: sectors / invariants]
P4 is a conserved sector label: the support decomposes into disconnected components (block
structure), or equivalently $P$ is block diagonal up to permutation, so evolution preserves a
sector coordinate.
\end{definition}

\begin{definition}[P5: packaging]
P5 is an idempotent endomap $e$ whose fixed points $\mathrm{Fix}(e)$ are the packaged objects of a
given theory. This is an endomap notion (Section~\ref{sec:idempotent-endo}).
It is not an order-closure.
\end{definition}

\begin{definition}[P6: accounting / audit]
P6 is an accounting/audit structure: a certificate or functional that is monotone under
coarse maps or packaging. Canonical instantiations in this manuscript include:
\begin{enumerate}
  \item the information/feasibility order induced by limited access (META),
  \item path-space KL asymmetry with data processing\\
  (Section~\ref{sec:aot}),
  \item the ACC graph 1-form/cycle-integral audit\\
  (Section~\ref{sec:acc}).
\end{enumerate}
\end{definition}

\begin{remark}
We write \emph{P6\_drive} for the ACC specialization: a non-exact log-ratio 1-form
(equivalently, a nonzero cycle integral). This is the condition used in the slogan
``P3 needs P6\_drive'' under autonomy.
\end{remark}

\subsection{How the primitives compose to generate theory growth}\label{sec:six-birds-loop}

In this manuscript a \emph{theory} is the package (lens/definability, completion, audit)
formalized in Section~\ref{sec:tk-theory-package}. Emergence means stable fixed points of the
completion rule inside a theory, while open-endedness means strict theory extension rather than
iteration of a fixed closure (Section~\ref{sec:forcing}).

\begin{enumerate}
  \item \textbf{Limited access $\Rightarrow$ P5 packaging.} A lens collapses distinctions, so a completion
  endomap packages microstates into fixed-point objects (Sections~\ref{sec:closure} and \ref{sec:idempotent-endo}).
  \item \textbf{Lossy packaging $\Rightarrow$ P6 accounting.} Audits/monotones track what survives under lenses
  (Sections~\ref{sec:acc} and \ref{sec:aot}); coarse-graining cannot create directionality by DPI.
  \item \textbf{Saturation $\Rightarrow$ extension.} Iterating a fixed completion saturates
  (Section~\ref{sec:closure}), so strict growth requires theory extension (Section~\ref{sec:forcing}).
  \item \textbf{P4 staging.} A refinement family supplies a theory index (Section~\ref{sec:meta-unavoidable});
  bounded interfaces/quantized indices make repeated refinement coherent.
  \item \textbf{P2 gating.} Feasibility carves representable macrostates and restricts support/cycle space
  (see the P2 theorem below and Section~\ref{sec:acc}).
  \item \textbf{P1 rewrite.} When induced macro-dynamics fails to descend, an operator rewrite/refinement is
  forced to restore closure on packaged objects (Section~\ref{sec:idempotent-endo}).
  \item \textbf{P3 route mismatch.} Noncommuting reductions yield route mismatch $\mathrm{RM}(j)$
  (Appendix~\ref{app:tk-defects}, Eq.~\eqref{eq:tk-rm}); this is route dependence, not a directionality certificate
  (Section~\ref{sec:aot}).
  \item \textbf{Iteration.} The updated (lens, completion, audit) defines the next theory; objects and audits
  change and the loop repeats.
\end{enumerate}

\paragraph{Mapping to the spine.}
\begin{itemize}
  \item P5 $\leftrightarrow$ completion/idempotents (Sections~\ref{sec:closure}, \ref{sec:idempotent-endo}).
  \item P6 $\leftrightarrow$ audits (Sections~\ref{sec:acc}, \ref{sec:aot}).
  \item P1/P2/P4 $\leftrightarrow$ closure-changing structure (this section and the P1/P2 theorems below).
  \item P3 $\leftrightarrow$ route mismatch (Appendix~\ref{app:tk-defects}).
  \item Open-endedness $\leftrightarrow$ strict theory extension (Section~\ref{sec:forcing}).
\end{itemize}

We do not claim every system self-organizes; the toolkit identifies the minimal operations and
certificates a system must instantiate if it exhibits autonomous theory growth and stable emergent objects.

\subsection{Downward influence across theories}\label{sec:downward-influence}

We make explicit how the emergence calculus represents \emph{downward influence} (macro-to-micro constraint) within a coupled theory package.
In this paper, ``downward influence'' is not a seventh primitive and not a metaphysical claim: it is a \emph{structural} phenomenon that appears whenever a coarse description is fed back into a finer one through the existing primitives (Section~\ref{sec:tk-theory-package}).

\paragraph{Three canonical downward mechanisms already in P1--P6.}
At the abstraction level of this paper, every downward influence path factors through one of the following constructions (the rest are compositions/specializations):

\begin{enumerate}
  \item \textbf{Representative selection via lift/completion (P5).}
  A lens $f:Z\to X$ induces a pushforward $Q_f$ and, once prototypes are fixed, a canonical lift $U_f:\Delta(X)\to\Delta(Z)$ (Definition~\ref{def:D-IC-01}).
  The map $x\mapsto u_x := U_f(\delta_x)$ is an explicit macro-to-micro channel: a coarse description selects a representative micro-description.
  This is exactly the ``completion'' move used to define dynamics-induced packaging maps such as $E_{\tau,f}$.

  \item \textbf{Feasibility gating and budgets (P2/P6).}
  A coarse constraint $C\subseteq X$ restricts the finer theory to the preimage $f^{-1}(C)\subseteq Z$ (and similarly for path constraints).
  Equivalently: feasibility/accounting at the coarse theory can \emph{exclude} micro directions by shrinking the admissible set of states/histories, without modifying the underlying micro update rule.

  \item \textbf{Internal protocol/phase as state (P3 with the autonomy guardrail).}
  A coarse context variable (``phase'') can be included as part of the state and can index a family of micro updates.
  This yields genuine context-dependent micro dynamics \emph{without} smuggling in an external schedule; see the protocol-trap discussion in Section~\ref{sec:aot}.
  Route mismatch/holonomy (P3) then measures noncommutativity of routes, but remains \emph{agnostic} about directionality unless coupled to a driven/asymmetric accounting structure (P6).
\end{enumerate}

\paragraph{Non-claim.}
This subsection does not assert that coarse descriptions ``override'' micro laws.
It only records that once a coarse theory is fed back into the finer one via (i) completion/lift, (ii) feasibility gating, or (iii) an internalized protocol variable, the resulting coupled theory package has bona fide macro-to-micro influence \emph{in the sense of constrained admissibility and context-indexed operators}.

\subsection{Two load-bearing propositions}

\paragraph{Constraints Kill Engines.}
\begin{theorem}[P2 gating shrinks cycle space]
\textbf{T-P2-01.}\par
Assume \textbf{A\_FIN}+\textbf{A\_REV}.\par
Let $G=(V,E)$ be an undirected support graph and let $G'=(V,E')$ be obtained by deleting edges.
Then
\[
\begin{aligned}
\beta_1(G')&\le \beta_1(G),\\
\beta_1(G)&:=|E|-|V|+c(G),
\end{aligned}
\]
where $c(G)$ is the number of connected components.
\end{theorem}

\begin{proof}
Deleting a single edge reduces $|E|$ by one and increases $c(G)$ by at most one. Hence
$|E|-|V|+c(G)$ cannot increase. Iterating over all deleted edges yields the claim.
\end{proof}

\begin{theorem}[P1 can change cycle rank and spectral gap (T-P1-01)]
Assume \textbf{A\_FIN}+\textbf{A\_REV}.
P1 rewrites can increase or decrease cycle rank and can increase or decrease a metastability proxy
(the spectral gap of a reversible lazy random walk).
\end{theorem}

\begin{proof}[Proof sketch]
(Cycle rank.) Adding an edge within an existing connected component leaves $c(G)$ unchanged and
increases $|E|$ by one, so $\beta_1(G)$ increases by one. More general rewires can change $c(G)$
and thus increase or decrease $\beta_1$.

(Spectral gap.) By rewriting edge weights or topology, one can create or remove bottlenecks.
Highly connected graphs have large gaps, while nearly decomposable graphs have small gaps; thus
P1 can increase or decrease the gap by explicit constructions (see Appendix~C for examples).
\end{proof}

\begin{remark}[Evidence pointers]
Appendix~C records finite graph and kernel witnesses for the P1 and P2 behaviors described above,
including explicit cycle-rank changes and spectral-gap increases/decreases.
\end{remark}

% Scope IDs: see docs/spec/theorem-inventory.md and paper-contract.md
\section{Examples}
\label{sec:examples}

This section gives a few finite examples that instantiate the preceding definitions and results.
They introduce no new machinery and are included purely for intuition and for checking that the
assumption tags match the intended use-cases.

\subsection{A biased three-cycle: a non-exact graph 1-form}
\label{subsec:ex:3cycle}

Let $Z=\{0,1,2\}$ and consider the Markov kernel $P$ defined by the parameters $(p,q,s)$ with
$p,q,s>0$ and $p+q+s=1$, where from each $i\in Z$,
\[
P(i,i)=s,\qquad P(i,i+1\bmod 3)=p,\qquad P(i,i-1\bmod 3)=q.
\]
This satisfies \textbf{A\_FIN}+\textbf{A\_REV}. The antisymmetric edge 1-form
$a_{ij}:=\log\!\big(P_{ij}/P_{ji}\big)$ is therefore defined on each bidirected edge.

Along the directed cycle $\gamma=(0\to 1\to 2\to 0)$, the cycle integral (affinity) is
\[
\oint_{\gamma} a
= \log\frac{P_{01}}{P_{10}}+\log\frac{P_{12}}{P_{21}}+\log\frac{P_{20}}{P_{02}}
= 3\log\frac{p}{q}.
\]
For instance with $(p,q,s)=(0.7,0.2,0.1)$ we get $|\oint_\gamma a|\approx 3.7583\neq 0$, hence the 1-form
is not exact and no global potential $\Phi$ can satisfy $a_{ij}=\Phi_j-\Phi_i$ on all edges (cf.\ Section~\ref{sec:acc}).

\subsection{Protocol trap: external schedule vs autonomous lifted model}
\label{subsec:ex:protocol-trap}

We illustrate the difference between an externally scheduled protocol and an autonomous lifted model.
Let $Z=\{0,1,2\}$ and let $K_0,K_1$ be two kernels that are each reversible w.r.t.\ the uniform distribution
(on $Z$), but do not commute (so the composition depends on order). One concrete instance used in our
reproducibility artifacts is:
\[
\begin{aligned}
K_0&=
\begin{pmatrix}
0.573333 & 0.393333 & 0.033333\\
0.393333 & 0.573333 & 0.033333\\
0.033333 & 0.033333 & 0.933333
\end{pmatrix},\\
K_1&=
\begin{pmatrix}
0.933333 & 0.033333 & 0.033333\\
0.033333 & 0.573333 & 0.393333\\
0.033333 & 0.393333 & 0.573333
\end{pmatrix}.
\end{aligned}
\]
for which $\max_{i,j}|(K_0K_1-K_1K_0)_{ij}|\approx 1.296\times 10^{-1}$.

\emph{Externally scheduled (not autonomous):}
if one applies $K_0$ then $K_1$ in a fixed periodic order but hides the schedule variable, the
stroboscopic one-step kernel on $Z$ is $K:=K_1K_0$, which can exhibit positive steady-state entropy
production (positive time-reversal asymmetry in the one-step Markov description).

\emph{Autonomous lifted model (clock included in state):}
define an augmented state space $Z\times \Phi$ where $\Phi=\{0,1\}$ is a phase variable, and define a
time-homogeneous random-scan Markov chain that, at each step, (i) updates $\phi$ using $\mathsf S$ while keeping
$z$ fixed with probability $\alpha$, and (ii) updates $z$ using $K_\phi$ while keeping $\phi$ fixed with
probability $1-\alpha$. If $\mathsf S$ is unbiased/reversible (no phase affinity) and both $K_0,K_1$ are reversible
w.r.t.\ a common $\pi$, then the lifted model exhibits
vanishing steady-state entropy production (and hence vanishing path reversal asymmetry at stationarity),
whereas introducing a phase bias (nontrivial affinity in the $\Phi$-subsystem) yields strictly positive
entropy production. This is the clean mathematical form of the slogan “P3 needs P6\_drive” in autonomous models:
noncommuting actions alone do not create sustained directionality unless the protocol degree of freedom
itself is driven.

\subsection{Finite forcing count: definability is exponentially rare}
\label{subsec:ex:forcing-count}

Let $Z$ be a finite set with $|Z|=N$, and let a “theory” be represented by a partition $\Pi$ of $Z$
into $K$ blocks (equivalently a lens $f:Z\to X$ with $|X|=K$). A predicate $h:Z\to\{0,1\}$ is definable
from $\Pi$ iff it is constant on each block of $\Pi$, so the number of definable predicates is exactly $2^K$,
while the total number of predicates is $2^N$. Therefore, if $h$ is sampled uniformly from $\{0,1\}^Z$,
\[
\mathbb P(h\text{ is definable from }\Pi)=\frac{2^K}{2^N}=2^{-(N-K)}.
\]
For example, when $(N,K)=(16,4)$, this gives $2^{-(16-4)}=2^{-12}\approx 2.44\times 10^{-4}$, i.e.\
a uniformly random new predicate is overwhelmingly likely to be non-definable from the old theory.
This is the finite counting core behind the “generic extension” step in Section~\ref{sec:forcing}.

\section{Discussion and scope boundary}
\subsection{What the theory does and does not claim}\label{sec:discussion-claims}

The organizing loop in Section~\ref{sec:big-picture} separates three roles that are often conflated:
packaging (idempotence and fixed points) and extension (non-definability), versus directionality (path-space KL and affinities).
The paper’s contribution is to make these roles explicit and composable, with certificates that are stable under coarse observation (DPI) and robust to hidden-clock artifacts (protocol audit).

In this paper, ``theory'' is used in the minimal mathematical sense of a definability structure
(lens/partition) together with its completion/packaging map. Later interpretations are intentionally out of scope here.
Viewed through Theorem~\ref{thm:meta-prim}, the primitives are not hypotheses about the world but structural consequences of description under bounded interfaces.
Instantiating the framework amounts to specifying a lens/refinement family, a completion rule (packaging), and an audit functional.
One then verifies that the domain’s interfaces realize the standing hypotheses.
None of these domain choices are fixed by the abstract theorem.

Equally important are the non-claims.
We do not claim that closure ladders arise automatically in arbitrary dynamics, nor that any particular choice of lens $f$ or timescale $\tau$ is preferred.
We also do not claim that protocol holonomy by itself yields sustained directionality under autonomy; on the contrary, the lifted-model statements formalize that sustained asymmetry requires a genuine affinity (a non-exact 1-form component) or an externally imposed schedule (which lies outside \textbf{A\_AUT}).
Finally, the finite forcing lemma is used only as a \emph{finite proxy} for generic extension: it justifies that strict theory extension is cheap and typical when there is hidden volume, not that all forms of novelty reduce to random predicates.

This paper keeps its scope finite, structural, and diagnostic: it does not advance claims about
continuous-time stochastic thermodynamics, empirical estimation guarantees, or domain-specific
applications. Appendix~D frames the Zeno decision frontier as a research outlook for when Zeno can or
cannot be ruled out in the abstract bridge setting.

\subsection{Outlook: forthcoming instantiations}\label{sec:outlook}
This paper isolates a math-only emergence calculus.
It uses a \emph{theory package} (lens/definability, completion, audit).
It is paired with the six primitive moves that change what theories can exist and what objects they recognize.
A series of companion papers will instantiate this calculus in multiple domains (with domain-specific choices of lenses, bridges, and audits) while keeping the present framework fixed:
\begin{itemize}
  \item mathematics theories,
  \item physics theories as families of effective closures across scales,
  \item life and cognition as theory stacks,
  \item conscious entities as theories,
  \item societies and civilizations,
  \item the universe as a “theorist”.
\end{itemize}
We intentionally do not include details here; the purpose of this paragraph is only to indicate scope and to delineate the boundary between the general calculus proved in this paper and its forthcoming instantiations.

\appendix

\section{Appendix A: Reproducibility checklist}
This repository is intended to be mechanically checkable at three levels:
(i) a claim/dependency registry, (ii) a small Lean formalization core, and
(iii) a deterministic Python evidence harness.
\par\noindent Repository: \url{https://github.com/ioannist/six-birds-theory}.\\
\noindent Arweave ID: 1kWCQsvq071tmOBx-uOVQfnsfsJZeHXhK7AWx-P-pII

\subsection*{A.1 Repository integrity checks (from repo root)}
Run the following commands from the repository root:
\begin{verbatim}
python3 scripts/check_deps_dag.py
python3 scripts/check_kb_pointers.py
python3 scripts/check_paper_contract.py
./check_lean.sh
\end{verbatim}

\subsection*{A.2 Python evidence harness (deterministic tests)}
The Python code lives under \texttt{python/} and is run in an isolated virtual environment.
From the repository root:
\begin{verbatim}
python3 -m venv python/.venv
python/.venv/bin/pip install -q pytest numpy
python/.venv/bin/pip install -e python
cd python && .venv/bin/python -m pytest -q
\end{verbatim}

\subsection*{A.3 Notes on tool availability}
A \LaTeX{} toolchain is not required to run the checks above.
If you wish to compile the PDF locally, install one of \texttt{latexmk}, \texttt{pdflatex}, or \texttt{tectonic} and run:
\begin{verbatim}
./scripts/build_paper.sh
\end{verbatim}

\section{Appendix B: Lean formalization map}
The Lean development provides a minimal mechanized backbone for the order-theoretic and
``thin packaging'' parts of the theory. It is intentionally small: no probability theory
or Markov chain results are mechanized in Lean for this paper.

\subsection*{B.1 What is mechanized}
The following items are implemented in Lean under \texttt{formal/}:

\begin{itemize}
  \item \textbf{Order-closure and closure ladders} (order-theoretic closure operators, fixed points, strict ladder relation, and iterate stabilization).
  \item \textbf{Thin packaging / reflection} for closure operators, expressed via the Galois insertion associated to a closure operator.
  \item \textbf{Idempotent endomaps} and fixed-point subtypes.\\
  Bridge lemma from order-closures to idempotent endomaps.
\end{itemize}

\subsection*{B.2 File map and key declarations}
\begin{itemize}
  \item \path{formal/ClosureLadder/Basic.lean}:\\
    closure operators, closed subtype, strict ladder relation, and iterate stabilization.\\
    Key lemmas used in the manuscript include \texttt{ladder\allowbreak *mono},
    \texttt{closure\allowbreak *iterate\allowbreak *succ}, and
    \texttt{closure\allowbreak *iterate\allowbreak *ge\allowbreak *one}.
  \item \path{formal/ClosureLadder/Packaging.lean}:\\
    reflection-style lemmas for the closed subtype (thin-category packaging anchor).
  \item \texttt{formal/ClosureLadder/}\\
    \texttt{IdempotentEndo.lean}:\\
    the structure \texttt{IdempotentEndo} and its fixed-point subtype.\\
    Bridge lemma: \texttt{ClosureOperator.}\\
    \texttt{toIdempotentEndo}.
  \item \path{formal/ClosureLadder.lean}:\\
    umbrella import module.
\end{itemize}

\subsection*{B.3 How to build}
From the repository root:
\begin{verbatim}
./check_lean.sh
\end{verbatim}
This script performs a cache fetch (when available) and then runs \texttt{lake build}.

\section{Appendix C: Python evidence map}
The repository contains a small deterministic Python harness (finite-state Markov kernels,
path laws, KL divergence, coarse-graining pushforwards, graph cycle utilities, and related
diagnostics). These computations are \emph{evidence and sanity checks} for the mathematical
definitions and theorem statements; they are not used as premises in any proof.

\subsection*{C.1 How to run}
See Appendix~A.2 for a one-shot test run via \texttt{pytest}. All tests are deterministic
(fixed seeds or no randomness).

\subsection*{C.2 Evidence by theme (tests and scripts)}
The following tests and scripts witness the main diagnostic phenomena discussed in the paper:

\paragraph{Path reversal asymmetry and data processing (DPI).}
Tests include randomized DPI checks, strict hiding examples, and nonstationary initial distributions.
Representative files:
\begin{verbatim}
python/tests/test_paths_kl_dpi.py
python/tests/test_paths_dpi_randomized.py
python/tests/test_paths_dpi_nonstationary_init.py
python/tests/test_paths_strict_hiding.py
python/scripts/dpi_sweep.py
\end{verbatim}

\paragraph{Protocol trap and the ``P3 needs P6\_drive'' correction under autonomy.}
Tests separate externally scheduled noncommuting kernels from autonomous phase-included models,
including m-phase protocols and the ``protocol clock audit''.
Representative files:
\begin{verbatim}
python/tests/test_protocol_trap_external_vs_autonomous.py
python/tests/test_protocol_clock_audit_phi_included.py
python/tests/test_protocol_mphase_random_scan.py
python/tests/test_protocol_stroboscopic.py
python/scripts/protocol_trap_demo.py
\end{verbatim}

\paragraph{Graph topology effects of P2 (edge deletion) and P1 (rewrites).}
Tests verify cycle-rank monotonicity under edge deletion, show cycle-rank increases under edge addition,
and demonstrate spectral-gap changes (including a decisive gap-decrease rewrite).
Representative files:
\begin{verbatim}
python/tests/test_graph_p2_cycle_rank_monotone.py
python/tests/test_graph_p1_edge_add_beta1_delta.py
python/tests/test_p1_spectral_gap_decrease.py
python/scripts/graph_p1_p2_demo.py
python/scripts/p1_gap_decrease_demo.py
\end{verbatim}

\paragraph{Dynamics-induced empirical endomap and idempotence defect.}
Tests implement the induced empirical endomap, the total-variation idempotence defect, and show that
refinement can help or hurt depending on leakage and timescale.
Representative files:
\begin{itemize}
  \item \texttt{python/tests/}\\
  \texttt{test\_empirical\_closure\_}\\
  \texttt{idempotence\_defect\_tv.py}
  \item \texttt{python/tests/}\\
  \texttt{test\_empirical\_closure\_}\\
  \texttt{refinement\_reveals\_objects.py}
  \item \texttt{python/tests/}\\
  \texttt{test\_empirical\_closure\_}\\
  \texttt{refinement\_can\_hurt.py}
  \item \path{python/scripts/empirical_closure_demo.py}
  \item \path{python/scripts/refinement_help_hurt_sweep.py}
\end{itemize}

\paragraph{Finite forcing / definability rarity.}
Tests include exact enumeration for small cases, Monte Carlo checks, and a contrast with fixed-weight sampling.
Representative files:
\begin{verbatim}
python/tests/test_forcing_exact_enumeration.py
python/tests/test_forcing_monte_carlo_matches_theory.py
python/tests/test_forcing_fixed_weight_differs.py
python/scripts/forcing_sweep.py
\end{verbatim}

\paragraph{Finite-memory ICAP witnesses.}
Tests implement the discrete-time finite-memory mass bound for truncated inputs.
Representative files:
\begin{verbatim}
python/src/clt/bridge_kernels.py
python/tests/test_bridge_kernels_fmem_icap.py
\end{verbatim}

\paragraph{Balanced-atom route (definitions + kernel-mass hinge).}
These files provide the finite-memory ICAP and toy compression witnesses used by the balanced-atom route.
Representative files:
\begin{verbatim}
python/src/clt/bridge_kernels.py
python/tests/test_bridge_kernels_fmem_icap.py
python/src/clt/balanced_atoms.py
python/tests/test_balanced_atoms_icap.py
python/scripts/balanced_atoms_witness.py
python/src/clt/ect_toy.py
python/tests/test_ect_compression_witness.py
python/scripts/ect_compression_witness.py
\end{verbatim}

\paragraph{ECT compression and capacity witnesses.}
Tests and scripts witness linear mode count and linear ICAP capacity growth in a toy family.
Representative files:
\begin{verbatim}
python/src/clt/ect_toy.py
python/tests/test_ect_compression_witness.py
python/scripts/ect_compression_witness.py
\end{verbatim}

\paragraph{ACC as graph 1-form exactness and cycle affinities.}
Tests compute the antisymmetric log-ratio 1-form, cycle integrals, and potential reconstruction when exact.
Representative files:
\begin{verbatim}
python/tests/test_acc_graph_one_form.py
\end{verbatim}

\section{Appendix D: Zeno cascades and depth}\label{app:zeno}
\subsection{Setup: frontier and Zeno criterion}
Fix a threshold $\theta>0$ and a scale index $j\in\mathbb N$.
Let $E_j:[0,\infty)\to[0,\infty)$ be an activity observable at level $j$, and define the frontier
\[
j_*(t):=\max\{j:\ E_j(t)\ge \theta\}.
\]
Define crossing times and step durations by
\[
t_j := \inf\{t:\ j_*(t)\ge j\},\qquad \Delta t_j := t_{j+1}-t_j.
\]
A \emph{Zeno cascade} is the event that $(t_j)$ remains bounded as $j\to\infty$, equivalently
$\sum_{j\ge j_0}\Delta t_j<\infty$.

\subsection{Bridges, ports, and accounting}
For each boundary $j$ introduce a real Hilbert space $U_j$ of port variables.
A port input is a signal $u_j:[0,\infty)\to U_j$ (measurable), and the eliminated sector is packaged
as a causal bridge operator $Z_j$ producing a port response $y_j=Z_j[u_j]$.
Define instantaneous interface power and positive supply by
\[
p_j(t):=\langle u_j(t),y_j(t)\rangle_{U_j},\qquad p_j^+(t):=\max\{p_j(t),0\},
\]
and the supplied work over an interval by
\[
W_j^+[s,t]:=\int_s^t p_j^+(\tau)\,d\tau.
\]
Using $W^+$ rather than net work avoids cancellations from backscatter and matches the use of
``forward supply'' in the latency argument.

\subsection{Storage-based activity and the WORK quantum (Option B)}
A bridge $Z_j$ is \emph{passive with storage} if there exists a nonnegative storage functional
$S_j(\cdot)\ge 0$ on an internal state $x_j(t)$ such that for all $s<t$,
\[
S_j(x_j(t)) - S_j(x_j(s))\ \le\ \int_s^t p_j(\tau)\,d\tau\ \le\ W_j^+[s,t].
\tag{$\ast$}\label{eq:passive-storage}
\]
Define the (storage-based) activity at level $j+1$ relative to the activation time $t_j$ by
\[
E_{j+1}(t):=
\begin{cases}
0, & t<t_j,\\[2mm]
\displaystyle \sup_{s\in[t_j,t]}\big(S_j(x_j(s))-S_j(x_j(t_j))\big), & t\ge t_j.
\end{cases}
\]
Define the next crossing time by
\[
t_{j+1}:=\inf\{t\ge t_j:\ E_{j+1}(t)\ge \theta\}.
\]
This definition is robust to ``pre-loading'' before $t_j$ because the baseline resets at $t_j$.

\begin{lemma}[WORK from storage growth]\label{lem:work-quantum}
Assume that the passive-storage inequality \eqref{eq:passive-storage} holds at boundary $j$,
and that $t\mapsto E_{j+1}(t)$ is right-continuous.
If $t_{j+1}<\infty$, then the supplied work over the crossing interval obeys
\[
W_j^+[t_j,t_{j+1}] \ \ge\ \theta.
\]
\end{lemma}
\begin{proof}
By right-continuity and the definition of $t_{j+1}$, $E_{j+1}(t_{j+1})\ge \theta$.
For any $\varepsilon>0$, choose $s_\varepsilon\in[t_j,t_{j+1}]$ with
$S_j(x_j(s_\varepsilon))-S_j(x_j(t_j))\ge \theta-\varepsilon$.
Apply \eqref{eq:passive-storage} on $[t_j,s_\varepsilon]$ to get
$W_j^+[t_j,s_\varepsilon]\ge \theta-\varepsilon$, and by monotonicity of $W_j^+$ in the interval,
$W_j^+[t_j,t_{j+1}]\ge \theta-\varepsilon$.
Letting $\varepsilon\to 0$ yields $W_j^+[t_j,t_{j+1}]\ge \theta$.
\end{proof}

\subsection{Integrated throughput (ICAP) and feasibility \texorpdfstring{$\Rightarrow$}{=>} latency bound}
The correct throughput notion for bridges with memory is an \emph{integrated} cap.
For all $s<t$, with truncated input $u_j^{[s,t]}:=u_j\,\mathbf 1_{[s,t]}$ and truncated output $y_j^{[s,t]}:=Z_j[u_j^{[s,t]}]$, define the windowed positive work
\[
W_{j,\mathrm{win}}^+[s,t]\ :=\ \int_s^t \max\{\langle u_j(\tau),y_j^{[s,t]}(\tau)\rangle,0\}\,d\tau.
\]
We say boundary $j$ satisfies an integrated capacity inequality (ICAP) if there exists $\Lambda(j)\in(0,\infty)$ such that
\[
W_{j,\mathrm{win}}^+[s,t]\ \le\ \Lambda(j)\int_s^t \|u_j(\tau)\|_{U_j}^2\,d\tau.
\tag{ICAP}\label{eq:icap}
\]
To convert this into a time lower bound we also isolate a primitive feasibility constraint:
there exists $\bar B(j)\in(0,\infty)$ such that on each crossing interval,
\[
\int_{t_j}^{t_{j+1}}\|u_j(\tau)\|_{U_j}^2\,d\tau\ \le\ \bar B(j)\,\Delta t_j.
\tag{FEAS}\label{eq:feas}
\]
This is an $L^2$ \emph{energy-density cap} (average squared port amplitude), which is compatible with memory and allows short spikes.

\begin{lemma}[Latency lower bound from WORK+ICAP+FEAS]\label{lem:latency-lb}
If Lemma~\ref{lem:work-quantum} holds at boundary $j$, \eqref{eq:icap}+\eqref{eq:feas} hold on $[t_j,t_{j+1}]$, and the following fresh-start alignment holds on the crossing interval,
\[
W_j^+[t_j,t_{j+1}] \ \le\ W_{j,\mathrm{win}}^+[t_j,t_{j+1}],
\tag{ALIGN}\label{eq:align}
\]
then
\[
\Delta t_j\ \ge\ \frac{\theta}{\Lambda(j)\,\bar B(j)}.
\]
\end{lemma}
\begin{proof}
By Lemma~\ref{lem:work-quantum}, $W_j^+[t_j,t_{j+1}]\ge \theta$.
By \eqref{eq:align}, $W_{j,\mathrm{win}}^+[t_j,t_{j+1}] \ge W_j^+[t_j,t_{j+1}] \ge \theta$.
By ICAP and FEAS,
\[
W_{j,\mathrm{win}}^+[t_j,t_{j+1}] \le \Lambda(j)\int_{t_j}^{t_{j+1}}\|u_j\|^2
\le \Lambda(j)\bar B(j)\Delta t_j.
\]
Rearrange.
\end{proof}
\begin{remark}[When the alignment hypothesis holds]
Condition \eqref{eq:align} is automatic for memoryless bridges and holds for memoryful bridges when the boundary is "fresh" at activation (e.g.\ the port is inactive before $t_j$ or the bridge state is reset at $t_j$).
\end{remark}

\subsection{No-Zeno criterion via divergence}
Define the effective per-boundary capacity $\mathrm{Cap}(j):=\Lambda(j)\bar B(j)$.
Combining Lemma~\ref{lem:latency-lb} with the elementary summability criterion yields a clean decision rule.

\begin{theorem}[No-Zeno from divergent reciprocal capacity]\label{thm:no-zeno}
Assume that for all sufficiently large $j$, Lemma~\ref{lem:latency-lb} applies with the same threshold $\theta>0$.
If
\[
\sum_{j\ge j_0}\frac{1}{\mathrm{Cap}(j)}\;=\;\infty,
\]
then $\sum_{j\ge j_0}\Delta t_j=\infty$ and hence no Zeno cascade occurs.
\end{theorem}
\begin{proof}
From Lemma~\ref{lem:latency-lb}, $\Delta t_j\ge \theta/\mathrm{Cap}(j)$, and the sum diverges by assumption.
\end{proof}

Appendix~E (Section~\ref{sec:ect-template}) provides one concrete route to certify DIV via mode compression and per-atom ICAP,
summarized in Theorem~\ref{thm:ect-summary}. This is optional evidence and does not change the formal criterion above.

\subsection{Hard lemma slots (WORK/CAP/route)}
The No-Zeno criterion above is elementary once three ingredients are available:
(i) a storage-based work quantum (Lemma~\ref{lem:work-quantum}),
(ii) an integrated throughput cap (ICAP), and
(iii) a feasibility/energy-density cap (FEAS).
In general systems, (ii)--(iii) are the true load-bearing frontier, and divergence is tied to how these
capabilities scale with depth; we isolate them as hard lemma slots.

\paragraph{HL-WORK (storage existence).}
For each $j$, there exists a storage functional $S_j\ge 0$ and internal state $x_j(t)$ such that \eqref{eq:passive-storage} holds.
\emph{Failure mode:} the packaged bridge is effectively active (storage can grow without supplied work).

\paragraph{HL-CAP-X1 (bounded dissipation density).}
Assuming a linear time-invariant (LTI) bridge, the positive-real constraint from passivity is supplemented by a uniform upper bound on the dissipative part, yielding ICAP \eqref{eq:icap}.
\emph{Failure mode:} arbitrarily sharp resonant dissipation (unbounded dissipative density) destroys any uniform $\Lambda(j)$.

\paragraph{HL-CAP-X2 (finite memory / kernel mass).}
Assuming a causal convolution form, a bound $\int_0^\infty\|k_j(t)\|\,dt<\infty$ yields an $L^2$ operator bound and hence ICAP.
\emph{Failure mode:} long memory tails (nonintegrable kernel mass) or singular instantaneous response.

\paragraph{HL-CAP-X3 (finite stable passive realization).}
Assuming a finite-dimensional passive state realization with exponential stability and uniformly bounded parameters, ICAP follows with an explicit $\Lambda(j)$ (up to an additive term depending on initial storage).
\emph{Failure mode:} marginally stable modes or scale-dependent complexity producing near-resonances.

\paragraph{HL-ROUTE (route mismatch controls gain growth).}
Let $\mathrm{RM}(j)$ quantify the mismatch between one-shot and two-step packaging across two levels.
A representative target inequality is
\[
\mathrm{Gain}(j\leftarrow j{+}2)\ \le\ \mathrm{Gain}(j\leftarrow j{+}1)\,\mathrm{Gain}(j{+}1\leftarrow j{+}2)\ +\ \mathrm{RM}(j),
\]
so small (e.g.\ summable) route mismatch prevents hidden super-gain from shortcut reductions.
\emph{Failure mode:} large route mismatch allows a one-shot packaged bridge to have gain far larger than sequential composition predicts.

\subsection{Checkable divergence criteria}
Write $\mathrm{Cap}(j)=\Lambda(j)\bar B(j)$.
The No-Zeno condition in Theorem~\ref{thm:no-zeno} is checkable from asymptotics.
For example, if there exist $C,c>0$ and exponents $\alpha,\beta\in\mathbb R$ such that for all large $j$,
\[
\Lambda(j)\le C(j+1)^\alpha,\qquad \bar B(j)\le C(j+1)^\beta,
\]
then $\mathrm{Cap}(j)\le C'(j+1)^{\alpha+\beta}$.
In particular,
\[
\sum_j \frac{1}{\mathrm{Cap}(j)}=\infty \qquad \text{whenever } \alpha+\beta\le 1,
\]
by comparison with the $p$-series.
More refined tests can be stated in terms of dyadic subsequences when $1/\mathrm{Cap}(j)$ is eventually nonincreasing.

\subsection{Toy model families (necessity witnesses)}
The criterion above is sharp in the following minimal sense.
If
\[
\sum_j \frac{1}{\mathrm{Cap}(j)}<\infty,
\]
then Zeno is compatible with WORK+ICAP+FEAS.

\paragraph{Zeno by fast capacity growth (DIV fails).}
Fix $\theta>0$ and define $\mathrm{Cap}(j)=2^j$.
Set $\Delta t_j:=\theta/\mathrm{Cap}(j)=\theta\,2^{-j}$ and $t_{j+1}:=t_j+\Delta t_j$.
Then $\sum_j\Delta t_j<\infty$.
Define $W_j^+[t_j,t_{j+1}]:=\theta$ and choose any $u_j$ with
$\int_{t_j}^{t_{j+1}}\|u_j\|^2=\Delta t_j$ so that FEAS holds with $\bar B(j)=1$.
ICAP holds with $\Lambda(j)=\mathrm{Cap}(j)$ by construction:
$W_j^+=\Lambda(j)\int\|u_j\|^2$.
This realizes Zeno while keeping the WORK quantum fixed.

\paragraph{Zeno by vanishing work quantum (WORK fails).}
Alternatively, even with a uniform capacity $\mathrm{Cap}(j)\equiv 1$, if the required work per level decays
(e.g.\ $\theta_j:=2^{-j}$), then taking $\Delta t_j=\theta_j$ yields $\sum_j\Delta t_j<\infty$.
This shows that some form of uniform work quantization is genuinely load-bearing for No-Zeno.

\subsection{Decision tree (settlement frontier)}
The No-Zeno engine decomposes into three independent settlement points:
\begin{itemize}
  \item \textbf{WORK (storage-based):} does the packaged bridge admit a passive storage inequality \eqref{eq:passive-storage}?
  \item \textbf{CAP (integrated):} can one derive an ICAP bound \eqref{eq:icap} (e.g.\ from passivity plus an additional structural condition)?
  \item \textbf{Divergence:} does $\sum_j 1/\mathrm{Cap}(j)$ diverge?
\end{itemize}
If all three hold (together with FEAS), Theorem~\ref{thm:no-zeno} rules out Zeno.
If any one fails, Zeno is not excluded in this abstract setting, and the failure mode indicates what structural assumption must break
(e.g.\ active bridges, uncontrolled resonance/memory, or rapidly growing effective capacity across scales).

\section{Appendix E: Toolkit theory---defects}\label{app:tk-defects}
\subsection{Defects as quantitative relaxations of exact laws}\label{sec:tk-defect-calculus}
% IDs: D-TK-DEF-01, D-TK-DEF-02, D-TK-ROUTE-01

A recurring pattern in this manuscript is that exact algebraic laws (idempotence, commutativity of reductions, audit monotonicity)
hold only approximately in realistic multiscale settings. We therefore standardize a \emph{defect calculus}:
each desired law is paired with a nonnegative scalar measuring its violation, and simple norm inequalities propagate these defects.

\paragraph{Idempotence defect (TV/L1).}
Let $Z$ be finite and let $\Delta(Z)$ denote the simplex of distributions on $Z$.
For a linear map $E:\Delta(Z)\to\Delta(Z)$ represented by a matrix acting on row distributions (so $\mu\mapsto \mu E$),
define the total-variation idempotence defect by
\[
\delta_{\mathrm{TV}}(E)\ :=\ \sup_{\mu\in\Delta(Z)}\big\|\mu(E^2-E)\big\|_{\mathrm{TV}}.
\tag{D-TK-DEF-02}\label{eq:tk-idem-def}
\]
For finite $Z$, $\|\nu\|_{\mathrm{TV}}=\tfrac12\|\nu\|_{1}$ for signed measures $\nu$, hence $\delta_{\mathrm{TV}}(E)$ is computable.

\begin{remark}[Extreme points attain the TV idempotence defect]\label{rem:tk-extreme-points}
Because $\mu\mapsto \|\mu(E^2-E)\|_{\mathrm{TV}}$ is convex in $\mu$, the supremum in \eqref{eq:tk-idem-def} is attained at extreme points:
\[
\delta_{\mathrm{TV}}(E)\ =\ \max_{z\in Z}\big\|\delta_z(E^2-E)\big\|_{\mathrm{TV}}
\ =\ \frac12\max_{i}\sum_{j}\big|(E^2-E)_{ij}\big|.
\]
This is the norm used by the reproducible diagnostics in the accompanying code (row-vector convention).
\end{remark}

\paragraph{Route mismatch defect (abstract).}
Many constructions in multiscale reduction depend on an elimination \emph{route}.
Let $\mathsf{Pack}_{j\leftarrow k}$ denote the effective packaging map from depth $k$ to depth $j$ (with $k>j$),
and let $d(\cdot,\cdot)$ be a metric on the space of packaged objects (e.g.\ an operator norm in linear settings).
Define the route mismatch defect across two steps by
\[
\mathrm{RM}(j)\ :=\ d\!\left(\mathsf{Pack}_{j\leftarrow j+2},\ \mathsf{Pack}_{j\leftarrow j+1}\circ \mathsf{Pack}_{j+1\leftarrow j+2}\right).
\tag{D-TK-ROUTE-01}\label{eq:tk-rm}
\]
This quantity is purely definitional here; in a later toolkit step we relate it to controlled growth of effective capacities.

\subsubsection*{Route dependence, capacity, and a one-step control inequality}\label{sec:tk-route-capacity}
% ID: L-TK-ROUTE-01

To make \eqref{eq:tk-rm} operational we isolate a minimal setting in which packaged objects are linear operators
between normed port spaces and the metric $d$ is induced by an operator norm. This covers, for example,
linearized bridges and bounded-gain summaries of more complex components.

\paragraph{Abstract packaging maps.}
For integers $k>j$, let $\mathsf{Pack}_{j\leftarrow k}$ denote the packaged object obtained by eliminating
intermediate structure between depths $k$ and $j$. We do not assume functoriality or strict associativity:
route dependence is precisely the possibility that
$\mathsf{Pack}_{j\leftarrow j+2}$ differs from $\mathsf{Pack}_{j\leftarrow j+1}\circ \mathsf{Pack}_{j+1\leftarrow j+2}$.

\paragraph{Linear-operator specialization.}\mbox{}\\
Assume each $\mathsf{Pack}_{j\leftarrow k}$ is a bounded linear operator between normed spaces,
with operator norm $\|\cdot\|$.\par
We take the metric in \eqref{eq:tk-rm} to be
\[
d(A,B):=\|A-B\|.
\]
Define the induced \emph{gain} of a packaged object by
\[
\mathrm{Gain}(j\leftarrow k)\ :=\ \|\mathsf{Pack}_{j\leftarrow k}\|.
\]

\begin{lemma}[Route mismatch controls gain growth]\label{lem:tk-route-capacity}
In the default linear-operator setting above, the two-step route mismatch
\[
\begin{aligned}
\mathrm{RM}(j)
&= \bigl\|\mathsf{Pack}_{j\leftarrow j+2}\\
&\quad - \mathsf{Pack}_{j\leftarrow j+1}\circ\mathsf{Pack}_{j+1\leftarrow j+2}\bigr\|
\end{aligned}
\]
implies the gain bound
\[
\begin{aligned}
\mathrm{Gain}(j\leftarrow j+2)
&\le \mathrm{Gain}(j\leftarrow j+1)\,\mathrm{Gain}(j+1\leftarrow j+2)\\
&\quad + \mathrm{RM}(j).
\end{aligned}
\tag{L-TK-ROUTE-01}\label{eq:tk-route-capacity}
\]
\end{lemma}
\begin{proof}
By the triangle inequality,
\[
\begin{aligned}
\|\mathsf{Pack}_{j\leftarrow j+2}\|
&\le \|\mathsf{Pack}_{j\leftarrow j+1}\circ\mathsf{Pack}_{j+1\leftarrow j+2}\|\\
&\quad + \|\mathsf{Pack}_{j\leftarrow j+2}-\mathsf{Pack}_{j\leftarrow j+1}\circ\mathsf{Pack}_{j+1\leftarrow j+2}\|.
\end{aligned}
\]
Submultiplicativity of the operator norm gives
\[
\|\mathsf{Pack}_{j\leftarrow j+1}\circ\mathsf{Pack}_{j+1\leftarrow j+2}\|
\le
\|\mathsf{Pack}_{j\leftarrow j+1}\|\,\|\mathsf{Pack}_{j+1\leftarrow j+2}\|.
\]
The final term is $\mathrm{RM}(j)$ by definition.
\end{proof}

\begin{remark}[Holonomy without teleology]\label{rem:tk-holonomy-no-teleology}
The purpose of $\mathrm{RM}(j)$ is to quantify \emph{route dependence} (a holonomy-like effect) of reduction/packaging.
This is logically distinct from directionality: a system may exhibit nontrivial route mismatch while having zero audit
(e.g.\ zero affinities or vanishing path reversal asymmetry under the stated assumptions).
The arrow-of-time results in Section~\ref{sec:aot} remain certified by the audit functional and its monotonicity,
not by route dependence alone.
\end{remark}

\subsection{Bridge objects: ports, passivity, and integrated throughput}\label{sec:tk-bridge-toolkit}
% IDs: D-TK-BRG-01, D-TK-BRG-02, D-TK-BRG-03

We isolate a minimal, representation-agnostic interface for \emph{bridge objects}:
a packaged eliminated sector acting causally on port variables of the retained sector.
The purpose is not to specialize to any particular physics, but to standardize what must be provided
to make WORK/CAP-style arguments and accounting statements well-posed.

\begin{definition}[Ports, causal bridge, and accounting]\label{def:tk-bridge-ports}
A \emph{port space} is a finite-dimensional real Hilbert space $U$ with inner product
$\langle\cdot,\cdot\rangle_U$ and norm $\|\cdot\|_U$.
Fix a finite time horizon $T>0$ and write $\mathcal U := L^2([0,T];U)$ for square-integrable $U$-valued signals.
A (deterministic) \emph{bridge operator} is a map $\mathsf Z:\mathcal U\to\mathcal U$.
It is \emph{causal} if for all $t\in[0,T]$ and all $u,v\in\mathcal U$,
\[
u(\tau)=v(\tau)\ \text{for a.e.\ }\tau\in[0,t]\quad\Longrightarrow\quad (\mathsf Z u)(\tau)=(\mathsf Z v)(\tau)\ \text{for a.e.\ }\tau\in[0,t].
\]
Given an input $u\in\mathcal U$ and output $y=\mathsf Z u$, define the instantaneous supply rate
\[
p(t):=\langle u(t),y(t)\rangle_U,\qquad p^+(t):=\max\{p(t),0\},
\]
and the positive supplied work over $[s,t]\subseteq[0,T]$ by
\[
W^+[s,t]:=\int_s^t p^+(\tau)\,d\tau.
\]
\end{definition}

\begin{definition}[Passivity with storage]\label{def:tk-passive-storage}
A causal bridge $\mathsf Z:\mathcal U\to\mathcal U$ is \emph{passive with storage} (cf.\ \cite{Willems1972}) if there exist:
(i) a state space $X$, (ii) a nonnegative storage functional $S:X\to[0,\infty)$,
and (iii) for each input $u\in\mathcal U$ a state trajectory $x_u:[0,T]\to X$,
such that for all $0\le s\le t\le T$,
\[
S(x_u(t)) - S(x_u(s))\ \le\ \int_s^t \langle u(\tau),(\mathsf Z u)(\tau)\rangle_U\,d\tau.
\tag{D-TK-BRG-02}\label{eq:tk-passivity}
\]
\end{definition}

\begin{definition}[Integrated capacity inequality (ICAP)]\label{def:tk-icap}
A causal bridge $\mathsf Z:\mathcal U\to\mathcal U$ satisfies an \emph{integrated capacity inequality} (ICAP)
if there exists $\Lambda\in(0,\infty)$ such that for all $0\le s\le t\le T$ and all inputs $u\in\mathcal U$,
with truncated input $u^{[s,t]}:=u\,\mathbf 1_{[s,t]}$ and output $y^{[s,t]}:=\mathsf Z(u^{[s,t]})$,
\[
W^+[s,t]\ \le\ \Lambda\int_s^t \|u(\tau)\|_U^2\,d\tau.
\tag{D-TK-BRG-03}\label{eq:tk-icap}
\]
ICAP is the natural throughput notion for bridges with memory: it allows short spikes in $p(t)$ while bounding
cumulative positive supply by an $L^2$ budget.
\end{definition}

\begin{remark}
The incremental (activated) form above avoids prehistory contamination in memory systems; a full-output per-subinterval ICAP
can fail even for simple delays.
\end{remark}

\begin{lemma}[Causality is closed under composition]\label{lem:tk-causal-compose}
Assume the setup of Definition~\ref{def:tk-bridge-ports}.
If $\mathsf Z_1:\mathcal U\to\mathcal U$ and $\mathsf Z_2:\mathcal U\to\mathcal U$ are causal, then
$\mathsf Z_2\circ \mathsf Z_1$ is causal.
\end{lemma}
\begin{proof}
If $u$ and $v$ agree a.e.\ on $[0,t]$, then by causality of $\mathsf Z_1$ the signals $\mathsf Z_1u$ and $\mathsf Z_1v$
agree a.e.\ on $[0,t]$. Applying causality of $\mathsf Z_2$ yields $(\mathsf Z_2\circ \mathsf Z_1)u=(\mathsf Z_2\circ \mathsf Z_1)v$
a.e.\ on $[0,t]$.
\end{proof}

\begin{lemma}[Parallel sum: passivity and ICAP constants add]\label{lem:tk-parallel-add}
Assume the setup of Definition~\ref{def:tk-bridge-ports}.
Let $\mathsf Z_1,\mathsf Z_2:\mathcal U\to\mathcal U$ be causal bridges on the same port space $U$, and define
$\mathsf Z:=\mathsf Z_1+\mathsf Z_2$ by $(\mathsf Zu)(t):=(\mathsf Z_1u)(t)+(\mathsf Z_2u)(t)$.
\begin{enumerate}
  \item If $\mathsf Z_1$ and $\mathsf Z_2$ are passive with storage (Definition~\ref{def:tk-passive-storage}),
  then $\mathsf Z$ is passive with storage (with storage equal to the sum of storages on the product state space).
  \item If $\mathsf Z_1$ and $\mathsf Z_2$ satisfy ICAP with constants $\Lambda_1$ and $\Lambda_2$, then
  $\mathsf Z$ satisfies ICAP with constant $\Lambda_1+\Lambda_2$.
\end{enumerate}
\end{lemma}
\begin{proof}
For (2), fix $u$ and $[s,t]$, and set $u^{[s,t]}:=u\,\mathbf 1_{[s,t]}$.
Let $y_i=\mathsf Z_i u^{[s,t]}$. Then $p=p_1+p_2$ pointwise.
Since $(a+b)^+\le a^+ + b^+$ for real $a,b$,
we have $W^+\le W_1^+ + W_2^+$, and applying ICAP to each term yields
\[
W^+[s,t]\le (\Lambda_1+\Lambda_2)\int_s^t \|u(\tau)\|_U^2\,d\tau.
\]
For (1), add the two passivity inequalities \eqref{eq:tk-passivity} on a product state space.
Use linearity of $\mathsf Z$ in the supply rate.
\end{proof}

\begin{remark}[Serial composition: gain submultiplicativity is model-free; passivity needs port-compatibility]\label{rem:tk-serial-note}
For bounded linear bridges on $\mathcal U$ one can define an operator-norm gain $\mathrm{Gain}(\mathsf Z):=\|\mathsf Z\|_{\mathcal U\to\mathcal U}$,
which satisfies $\mathrm{Gain}(\mathsf Z_2\circ\mathsf Z_1)\le \mathrm{Gain}(\mathsf Z_2)\mathrm{Gain}(\mathsf Z_1)$.
Such a bound yields an ICAP constant via Cauchy--Schwarz (so $\mathrm{Cap}(\mathsf Z)\le \mathrm{Gain}(\mathsf Z)$),
hence the resulting ICAP constants are submultiplicative under serial composition in this default linear setting.
By contrast, passivity under serial/feedback interconnection depends on how port variables are paired and how the external
supply decomposes into internal supplies; we therefore record the parallel-sum lemma above as the universally valid closure property
in the present one-port notation.
\end{remark}

\begin{remark}[Relation to primitives P5 and P6]\label{rem:tk-bridge-to-primitives}
In the main text, \textbf{P5} (packaging) is expressed as a completion endomap $E$ with fixed points,
while the present appendix expresses packaging at a boundary as a \emph{causal bridge} $\mathsf Z$ acting on port signals.
Both viewpoints represent ``the eliminated sector as an effective operator''; one acts on descriptions (e.g.\ distributions),
the other acts on boundary interactions.
The audit components used in this manuscript---path reversal asymmetry (Section~\ref{sec:aot}) and graph 1-form affinities (Section~\ref{sec:acc})---are
instances of \textbf{P6}-style accounting: they certify nontrivial drive/directionality and obey monotonicity principles under coarse observation.
\end{remark}

\subsection{Emergent coercivity template via sector compression}\label{sec:ect-template}

\subsubsection{Summary: slots and divergence consequence}\label{sec:ect-summary}

\begin{theorem}[ECT summary: linear ICAP capacity and DIV from three slots]\label{thm:ect-summary}
Fix a depth-indexed family of packaged bridges $(\mathsf Z_j)_j$ with port spaces $U_j$ and a horizon $T>0$.
Assume the following slot conditions hold for all sufficiently large $j$:

\begin{enumerate}
  \item[\textbf{HL-ECT-1}] (\emph{Mode compression}) $\mathsf Z_j$ admits an atomic/sector decomposition with mode count
  $m_j \le C_0 (j+1)$.
  A sufficient route in this appendix is: sectorization + P5-min + quantized index, which yields $m_j \lesssim (j+1)$
  (see Corollary~\ref{cor:ect-hl1}).
  \item[\textbf{HL-ECT-3}] (\emph{Uniform per-atom ICAP}) each atom $\mathsf Z_{j,r}$ in the decomposition satisfies ICAP with a
  uniform constant $\Lambda_0$ (independent of $j,r$). A sufficient condition is bounded kernel mass
  (Lemma~\ref{lem:ect-fmem}). In particular, dissipative atoms with semigroup decay (Definition~\ref{def:ect-atom-ss})
  and balanced coupling (Definition~\ref{def:ect-bal}) satisfy uniform ICAP via Lemma~\ref{lem:ect-bal-kmass}
  and Corollary~\ref{cor:ect-bal-icap}.
  \item[\textbf{HL-ECT-2}] (\emph{Feasibility removes lossless directions, if strictness is needed}) feasibility gating excludes
  lossless directions so that a coercive norm is induced on feasible inputs (see Remark~\ref{rem:hl-p2-core},
  Lemma~\ref{lem:ect-coer-feas}, and the depth-scaling lemma Lemma~\ref{lem:ect-shrink}).
\end{enumerate}

Then the ICAP capacity constant obeys a linear bound
\[
\mathrm{Cap}(\mathsf Z_j)\ \le\ \Lambda_0\,C_0\,(j+1),
\]
and consequently
\[
\sum_{j\ge 0}\frac{1}{\mathrm{Cap}(\mathsf Z_j)}=\infty.
\]
In particular, this verifies the \emph{DIV} ingredient used in the No-Zeno criterion (Theorem~\ref{thm:no-zeno} in Appendix~D).
WORK quantization and a feasibility/energy-density cap are supplied there.
\end{theorem}

\begin{proof}
Under \textbf{HL-ECT-1} and \textbf{HL-ECT-3}, Lemma~\ref{lem:ect-mech} gives $\mathrm{Cap}(\mathsf Z_j)\le m_j\Lambda_0 \le \Lambda_0 C_0 (j+1)$,
and the divergence of $\sum_j 1/\mathrm{Cap}(\mathsf Z_j)$ follows by comparison with the harmonic series.
\end{proof}

\begin{remark}[Non-implication: why ECT is an add-on, not automatic]\label{rem:ect-nonimplication}
The ECT conclusion is \emph{not} a consequence of passivity alone.
Passivity (existence of a storage functional with $S(t)-S(s)\le \int_s^t \langle u, \mathsf Z u\rangle$) does not by itself
control ICAP constants, does not force mode compression, and does not eliminate lossless feasible directions.
Even within convolution/finite-memory classes, passivity alone allows arbitrarily large ICAP constants via high-\emph{Q}
slowly decaying internal directions; the balanced coupling condition (Definition~\ref{def:ect-bal}) is one explicit way to exclude this.
The template therefore isolates explicit structural sources of coercivity/capacity control:
\textbf{P4}-style sector compression (bounded sector complexity), \textbf{P5} packaging outcomes (e.g.\ atomization with finite memory),
and \textbf{P2} feasibility that removes null directions. These conditions couple packaging+accounting to a throughput certificate,
but they do not by themselves imply novelty (Section~\ref{sec:forcing}) or directionality (Sections~\ref{sec:acc}--\ref{sec:aot}).
\end{remark}

% IDs: D-ECT-ATOM-01, D-ECT-CAP-01, L-ECT-MECH-01

\begin{definition}[Atomic/sector decompositions of a packaged bridge {\normalfont\textsf{(ID: D-ECT-ATOM-01)}}]\label{def:ect-atom}
Fix a port space $U$ and a finite horizon $T>0$, and write $\mathcal U:=L^2([0,T];U)$.
Let $\mathsf Z:\mathcal U\to\mathcal U$ be a causal bridge (Definition~\ref{def:tk-bridge-ports}).
An \emph{atomic (sector) decomposition} of $\mathsf Z$ is a representation
\[
\mathsf Z \;=\; \sum_{r=1}^m \mathsf Z_r
\]
as a parallel sum of causal bridges $\mathsf Z_r:\mathcal U\to\mathcal U$ on the same port space (so $(\mathsf Zu)(t)=\sum_r(\mathsf Z_ru)(t)$).
We call $m$ the \emph{mode count} of the chosen decomposition.
In depth-indexed settings we write $\mathsf Z=\mathsf Z_j$ and denote the mode count by $m_j$.
\end{definition}

\subsubsection{Dissipative atoms and semigroup decay}
\begin{definition}[Dissipative atom representation {\normalfont\textsf{(ID: D-ECT-ATOM-SS-01)}}]\label{def:ect-atom-ss}
Let $\mathsf Z_r$ be an atom in an atomic decomposition (Definition~\ref{def:ect-atom}) acting on
$\mathcal U:=L^2([0,T];U)$. We say $\mathsf Z_r$ admits a \emph{dissipative atom representation} if it has a kernel
\[
K_r(\tau)\ =\ C_r e^{A_r\tau} B_r,\qquad \tau\ge 0,
\]
with an internal Hilbert space $H_r$, bounded linear maps $A_r:H_r\to H_r$, $B_r:U\to H_r$, $C_r:H_r\to U$, and decay
\[
\|e^{A_r\tau}\|_{\mathrm{op}}\ \le\ e^{-\lambda_r \tau}\quad\text{for some }\lambda_r>0.
\]
All windowed/interval inequalities in this appendix are evaluated using truncated inputs
$u^{[s,t]}:=u\,\mathbf 1_{[s,t]}$, as in Definition~\ref{def:ect-cap}.
\end{definition}

\begin{definition}[Balanced coupling {\normalfont\textsf{(ID: D-ECT-BAL-01)}}]\label{def:ect-bal}
\mbox{}\par
In the setting of Definition~\ref{def:ect-atom-ss}, fix a constant $\Lambda_0>0$.
We say the atom $\mathsf Z_r$ satisfies the \emph{balanced coupling condition} with constant $\Lambda_0$ if
\[
\|C_r\|_{\mathrm{op}}\ \|B_r\|_{\mathrm{op}}\ \le\ \Lambda_0\,\lambda_r,
\]
where $\lambda_r>0$ is the semigroup decay rate from Definition~\ref{def:ect-atom-ss}.
\end{definition}

\begin{definition}[ICAP capacity functional {\normalfont\textsf{(ID: D-ECT-CAP-01)}}]\label{def:ect-cap}
For a causal bridge $\mathsf Z:\mathcal U\to\mathcal U$, define its \emph{ICAP capacity constant}
\[
\mathrm{Cap}(\mathsf Z)\ :=\ \inf\Big\{\Lambda\in(0,\infty):\
\begin{aligned}
W^+[s,t] &\le \Lambda\int_s^t\|u(\tau)\|_U^2\,d\tau\\
&\text{for all }0\le s\le t\le T \text{ and } u\in\mathcal U
\end{aligned}
\Big\},
\]
where $u^{[s,t]}:=u\,\mathbf 1_{[s,t]}$, $y^{[s,t]}:=\mathsf Z(u^{[s,t]})$, and $W^+[s,t]$ is computed from
$u^{[s,t]},y^{[s,t]}$ as in Definition~\ref{def:tk-bridge-ports}.
(Equivalently, $\mathrm{Cap}(\mathsf Z)$ is the best constant in \eqref{eq:tk-icap} when it holds.)
\end{definition}

\begin{lemma}[Balanced decay implies a uniform kernel-mass bound {\normalfont\textsf{(ID: L-ECT-BAL-KMASS-01)}}]\label{lem:ect-bal-kmass}
Assume $\mathsf Z_r$ admits a dissipative atom representation (Definition~\ref{def:ect-atom-ss})
and satisfies the balanced coupling condition with constant $\Lambda_0$ (Definition~\ref{def:ect-bal}).
Then the operator-valued kernel $K_r(\tau)=C_r e^{A_r\tau}B_r$ satisfies
\[
\int_0^\infty \|K_r(\tau)\|_{\mathrm{op}}\,d\tau\ \le\ \Lambda_0.
\]
\end{lemma}
\begin{proof}
By submultiplicativity and decay,
\[
\|K_r(\tau)\|_{\mathrm{op}}
\le \|C_r\|_{\mathrm{op}}\|e^{A_r\tau}\|_{\mathrm{op}}\|B_r\|_{\mathrm{op}}
\le \|C_r\|_{\mathrm{op}}\|B_r\|_{\mathrm{op}}e^{-\lambda_r\tau}.
\]
Integrating gives
\[
\int_0^\infty \|K_r(\tau)\|_{\mathrm{op}}\,d\tau
\le \frac{\|C_r\|_{\mathrm{op}}\|B_r\|_{\mathrm{op}}}{\lambda_r}
\le \Lambda_0,
\]
by Definition~\ref{def:ect-bal}.
\end{proof}

\begin{lemma}[Finite kernel-mass convolution bridges satisfy ICAP for truncated inputs {\normalfont\textsf{(ID: L-ECT-FMEM-01)}}]\label{lem:ect-fmem}
Let $U$ be a finite-dimensional port space and $\mathcal U:=L^2([0,T];U)$.
Consider a causal convolution bridge $\mathsf Z:\mathcal U\to\mathcal U$ of the form
\[
(\mathsf Zu)(t)\ :=\ \int_0^t K(t-\tau)\,u(\tau)\,d\tau,
\]
where $K:[0,\infty)\to \mathrm{End}(U)$ is measurable and $k(\tau):=\|K(\tau)\|_{\mathrm{op}}$ satisfies
$M:=\int_0^\infty k(\tau)\,d\tau<\infty$.
For $0\le s\le t\le T$ define the truncated input $u^{[s,t]}:=u\,\mathbf 1_{[s,t]}$ and output $y^{[s,t]}:=\mathsf Z(u^{[s,t]})$.
Then the positive supplied work computed from $y^{[s,t]}$ satisfies
\[
W^+[s,t]\ \le\ M\int_s^t \|u(\tau)\|_U^2\,d\tau.
\]
\end{lemma}
\begin{proof}
For $\tau\in[s,t]$ we have
\[
\|y^{[s,t]}(\tau)\|_U\ \le\ \int_s^\tau k(\tau-\sigma)\,\|u(\sigma)\|_U\,d\sigma,
\]
so by Young's inequality, $\|y^{[s,t]}\|_{L^2([s,t])}\le M\,\|u\|_{L^2([s,t])}$.
Then
\[
\begin{aligned}
W^+[s,t]
&\le \int_s^t \|u(\tau)\|_U\,\|y^{[s,t]}(\tau)\|_U\,d\tau\\
&\le \|u\|_{L^2([s,t])}\,\|y^{[s,t]}\|_{L^2([s,t])}\\
&\le M\int_s^t \|u(\tau)\|_U^2\,d\tau.
\end{aligned}
\]
\end{proof}

\begin{corollary}[Balanced atoms have uniform ICAP]\label{cor:ect-bal-icap}
\mbox{}\par
\textbf{C-ECT-BAL-ICAP-01.}\par
Assume $\mathsf Z_r$ admits a dissipative atom representation (Definition~\ref{def:ect-atom-ss})
and satisfies the balanced coupling condition with constant $\Lambda_0$ (Definition~\ref{def:ect-bal}).
Then $\mathrm{Cap}(\mathsf Z_r)\le \Lambda_0$ in the sense of Definition~\ref{def:ect-cap}.\par
That is, the ICAP inequality holds with constant $\Lambda_0$ for truncated inputs $u^{[s,t]}$
and positive work $W^+[s,t]$.
\end{corollary}
\begin{proof}
Lemma~\ref{lem:ect-bal-kmass} gives a kernel-mass bound $M\le \Lambda_0$.
Applying Lemma~\ref{lem:ect-fmem} with this $M$ yields
\[
W^+[s,t]\le \Lambda_0\int_s^t\|u(\tau)\|_U^2\,d\tau
\]
for all windows. Hence $\mathrm{Cap}(\mathsf Z_r)\le \Lambda_0$ by Definition~\ref{def:ect-cap}.
\end{proof}

\begin{remark}[Finite-memory bridges as a P5 packaging target]\label{rem:ect-p5-fmem}
Lemma~\ref{lem:ect-fmem} provides a concrete sufficient condition for the uniform per-atom ICAP hypothesis used in Lemma~\ref{lem:ect-mech}.
Accordingly, one natural packaging target for \textbf{P5} is to return atom bridges that admit a convolution representation with bounded kernel mass $M$,
so that $\mathrm{Cap}(\mathsf Z_r)\le M$ is explicit and checkable.
This is only a model class and not a universal consequence of packaging.
\end{remark}

\begin{lemma}[Mechanical ECT: sector ICAP + mode count $\Rightarrow$ linear capacity {\normalfont\textsf{(ID: L-ECT-MECH-01)}}]\label{lem:ect-mech}
Let $\mathsf Z=\sum_{r=1}^m \mathsf Z_r$ be an atomic decomposition in the sense of Definition~\ref{def:ect-atom}.
Assume each atom satisfies ICAP with a uniform constant $\Lambda_0$, i.e.\ for all $r$ and all $0\le s\le t\le T$,
\[
W_r^+[s,t]\le \Lambda_0\int_s^t\|u(\tau)\|_U^2\,d\tau,
\]
where $u^{[s,t]}:=u\,\mathbf 1_{[s,t]}$ and $W_r^+[s,t]$ is computed from $y_r=\mathsf Z_r u^{[s,t]}$.
Then $\mathsf Z$ satisfies ICAP with constant $m\Lambda_0$, hence
\[
\mathrm{Cap}(\mathsf Z)\ \le\ m\Lambda_0.
\]
In particular, in a depth-indexed family $(\mathsf Z_j)_j$ with mode counts $m_j\le C_0(j+1)$ we obtain
\[
\mathrm{Cap}(\mathsf Z_j)\ \le\ \Lambda_0\,C_0\,(j+1),
\qquad\text{and consequently}\qquad
\sum_{j\ge 0}\frac{1}{\mathrm{Cap}(\mathsf Z_j)}=\infty.
\]
\end{lemma}
\begin{proof}
Apply Lemma~\ref{lem:tk-parallel-add}(2) iteratively to the parallel sum
$\mathsf Z=\mathsf Z_1+\cdots+\mathsf Z_m$ to conclude that the ICAP constants add.
Thus $W^+[s,t]\le \sum_{r=1}^m W_r^+[s,t]\le m\Lambda_0\int_s^t\|u(\tau)\|_U^2\,d\tau$.
The bound on $\mathrm{Cap}(\mathsf Z)$ is immediate from Definition~\ref{def:ect-cap}.
If $\mathrm{Cap}(\mathsf Z_j)\le \Lambda_0C_0(j+1)$, then $\sum_{j\ge 0}1/\mathrm{Cap}(\mathsf Z_j)$ diverges by comparison with the harmonic series.
\end{proof}

\begin{remark}[DIV-ready chain]\label{rem:ect-div}
Fix a depth $j$ and interval $[t_j,t_{j+1}]$ with input $u_j$ and output $y_j=\mathsf Z_j u_j$.
If the positive supplied work satisfies $W_j^+[t_j,t_{j+1}]\ge \theta$ for some $\theta>0$, then by ICAP,
\[
\int_{t_j}^{t_{j+1}} \|u_j(\tau)\|_{U_j}^2\,d\tau \ \ge\ \frac{\theta}{\mathrm{Cap}(\mathsf Z_j)}.
\]
If, in addition, there is a bound
\[
\int_{t_j}^{t_{j+1}} \|u_j(\tau)\|_{U_j}^2\,d\tau \le B\,\Delta t_j,
\qquad \Delta t_j:=t_{j+1}-t_j,
\]
then $\Delta t_j \ge \theta/(B\,\mathrm{Cap}(\mathsf Z_j))$.
Under a linear capacity bound $\mathrm{Cap}(\mathsf Z_j)\le C_1(j+1)$, the sum $\sum_j \Delta t_j$ diverges by comparison with the harmonic series.
\end{remark}

\begin{remark}[Route mismatch as an additive gain error]\label{rem:ect-route}
In the default linear-operator specialization of Appendix~E, Lemma~\ref{lem:tk-route-capacity} and \eqref{eq:tk-route-capacity} show that
route mismatch contributes an additive term to two-step gain growth.
Thus, when estimating gains along successive reductions, $\mathrm{RM}(j)$ appears as an additive perturbation term.
When the mismatch operator itself satisfies ICAP with a comparable constant, this additive term can be absorbed into the same $\mathrm{Cap}(\cdot)$ calculus;
otherwise it remains an explicit robustness correction.
\end{remark}

\subsubsection{Coercivity from feasibility gating (P2)}\label{sec:ect-gating}
% IDs: D-ECT-GATE-01, L-ECT-COER-01

\begin{definition}[Feasibility gate and induced coercive norm {\normalfont\textsf{(ID: D-ECT-GATE-01)}}]\label{def:ect-gate}
Fix a finite horizon $T>0$, a port Hilbert space $U_j$, and write $\mathcal U_j:=L^2([0,T];U_j)$.
Let $G_j:U_j\to U_j$ be a self-adjoint positive semidefinite operator.
Define the feasible input class
\[
\mathcal U_j^{\mathrm{feas}}:=\{u\in\mathcal U_j:\ u(t)\in \mathrm{Ran}(G_j)\ \text{for a.e.\ }t\},
\]
the pointwise seminorm on $U_j$,
\[
|v|_{X_j}^2\ :=\ \langle v,G_j v\rangle\quad (v\in U_j),
\]
and the associated integrated energy on a window $[s,t]$,
\[
|u|_{L^2_X(s,t)}^2\ :=\ \int_s^t |u(\tau)|_{X_j}^2\,d\tau.
\]
On $\mathcal U_j^{\mathrm{feas}}$ this is a genuine norm, since the feasibility constraint removes the kernel of $G_j$.
\end{definition}

\begin{lemma}[Coercivity on the feasible set {\normalfont\textsf{(ID: L-ECT-COER-01)}}]\label{lem:ect-coer-feas}
Let $\mathsf Z_j:\mathcal U_j\to\mathcal U_j$ be a causal bridge.
Assume that for every feasible input $u\in\mathcal U_j^{\mathrm{feas}}$ and every $0\le s\le t\le T$, with output $y=\mathsf Z_j u$,
\[
\int_s^t \langle u(\tau),y(\tau)\rangle\,d\tau\ \ge\ \int_s^t \langle u(\tau),G_j u(\tau)\rangle\,d\tau.
\]
Then, on $\mathcal U_j^{\mathrm{feas}}$,
\[
\int_s^t \langle u(\tau),(\mathsf Z_j u)(\tau)\rangle\,d\tau\ \ge\ |u|_{L^2_X(s,t)}^2.
\]
\end{lemma}
\begin{proof}
The right-hand side is exactly $\int_s^t \langle u(\tau),G_j u(\tau)\rangle\,d\tau$ by definition of $|u|_{L^2_X(s,t)}^2$.
\end{proof}

\begin{remark}[Hard slot HL-P2-CORE (P2 must eliminate lossless feasible directions)]\label{rem:hl-p2-core}
\emph{Slot (HL-P2-CORE).} There exists a self-adjoint $G_j\succeq 0$ whose kernel coincides with the lossless directions of the packaged bridge at depth $j$
(e.g.\ in a linear specialization, $\ker(G_j)=\ker(\mathrm{Sym}\,\mathsf Z_j)$), and \textbf{P2} feasibility implies
$u(t)\in \mathrm{Ran}(G_j)$ for a.e.\ $t$.
Once this slot holds, Lemma~\ref{lem:ect-coer-feas} turns the gating statement into the strict/coercive inequality on feasible inputs.
All difficulty is concentrated in establishing HL-P2-CORE for the particular packaging construction at hand.
\end{remark}

\begin{lemma}[Feasibility shrinkage $\Rightarrow$ depth-scaled coercivity {\normalfont\textsf{(ID: L-ECT-SHRINK-01)}}]\label{lem:ect-shrink}
Assume the setup of Definition~\ref{def:ect-gate} and Lemma~\ref{lem:ect-coer-feas}.
Suppose there exists a constant $a_j>0$ such that for every feasible input $u\in\mathcal U_j^{\mathrm{feas}}$ and every interval $0\le s\le t\le T$,
\[
|u|_{L^2_X(s,t)}^2 \;\ge\; a_j^2\int_s^t \|u(\tau)\|_{U_j}^2\,d\tau.
\]
Then for every feasible $u$ and every $0\le s\le t\le T$ we have the depth-scaled coercive lower bound
\[
\int_s^t \langle u(\tau),(\mathsf Z_j u)(\tau)\rangle\,d\tau
\;\ge\;
a_j^2\int_s^t \|u(\tau)\|_{U_j}^2\,d\tau.
\]
Equivalently, relative to the baseline $L^2([0,T];U_j)$ norm the effective coercivity coefficient is $c_j=a_j^2$.
\end{lemma}
\begin{proof}
Combine Lemma~\ref{lem:ect-coer-feas} with the stated shrinkage inequality.
\end{proof}

\begin{remark}[Examples of depth scaling]\label{rem:ect-shrink-examples}
If $a_j \gtrsim (j+1)$, then $c_j=a_j^2 \gtrsim (j+1)^2$.
If $a_j \gtrsim 2^{s j}$, then $c_j=a_j^2 \gtrsim 2^{2 s j}$.
These are purely algebraic consequences of the shrinkage hypothesis.
\end{remark}

\subsubsection{Mode compression from sectorization (P4) and minimality (P5)}\label{sec:ect-sectorization}
% IDs: D-ECT-SECT-01, D-ECT-P5MIN-01, L-ECT-MODE-UB-01

\begin{definition}[Sectorization of the port space {\normalfont\textsf{(ID: D-ECT-SECT-01)}}]\label{def:ect-sectorization}
Fix a depth $j$ and a port Hilbert space $U_j$.
A \emph{sectorization} is the choice of a finite index set $Q_j$ and a family of orthogonal projections
$\{\Pi_{j,q}:U_j\to U_j\}_{q\in Q_j}$ such that $\Pi_{j,q}\Pi_{j,q'}=0$ for $q\ne q'$ and $\sum_{q\in Q_j}\Pi_{j,q}=I$.
We extend $\Pi_{j,q}$ pointwise to signals $u\in\mathcal U_j:=L^2([0,T];U_j)$ by $(\Pi_{j,q}u)(t):=\Pi_{j,q}(u(t))$.
A causal bridge $\mathsf Z_j:\mathcal U_j\to \mathcal U_j$ is \emph{sector-respecting} if
\[
\mathsf Z_j \;=\; \sum_{q\in Q_j} \Pi_{j,q}\,\mathsf Z_j\,\Pi_{j,q},
\]
i.e.\ the output on each sector depends only on the input restricted to that sector (no cross-sector coupling).
\end{definition}

\begin{definition}[Sector-respecting minimal decompositions (P5-min) {\normalfont\textsf{(ID: D-ECT-P5MIN-01)}}]\label{def:ect-p5min}
Assume a sectorization as in Definition~\ref{def:ect-sectorization}.
An atomic decomposition $\mathsf Z_j=\sum_{r=1}^m \mathsf Z_{j,r}$ (Definition~\ref{def:ect-atom}) is \emph{sector-respecting} if
for each atom $r$ there exists a sector label $\sigma(r)\in Q_j$ such that
\[
\mathsf Z_{j,r} \;=\; \Pi_{j,\sigma(r)}\,\mathsf Z_{j,r}\,\Pi_{j,\sigma(r)}.
\]
Define the \emph{sector-minimal mode count} $m_j^{\min}$ to be the smallest $m$ for which $\mathsf Z_j$ admits a
sector-respecting decomposition into $m$ \emph{nonzero} atoms.
We say packaging satisfies \textbf{P5-min} at depth $j$ if it returns a sector-respecting decomposition achieving $m_j^{\min}$.
\end{definition}

\begin{lemma}[Sectorization bounds the minimal mode count {\normalfont\textsf{(ID: L-ECT-MODE-UB-01)}}]\label{lem:ect-mode-ub}
Assume a sectorization as in Definition~\ref{def:ect-sectorization}.
Then any sector-respecting atomic decomposition of $\mathsf Z_j$ can be merged within each sector to produce a sector-respecting
decomposition with at most $|Q_j|$ nonzero atoms. In particular,
\[
m_j^{\min}\ \le\ |Q_j|.
\]
Consequently, under \textbf{P5-min} packaging (Definition~\ref{def:ect-p5min}) the returned mode count satisfies $m_j\le |Q_j|$.
\end{lemma}
\begin{proof}
Let $\mathsf Z_j=\sum_{r=1}^m \mathsf Z_{j,r}$ be sector-respecting with labels $\sigma(r)\in Q_j$.
For each $q\in Q_j$ define the merged atom $\widetilde{\mathsf Z}_{j,q}:=\sum_{\{r:\sigma(r)=q\}} \mathsf Z_{j,r}$.
Each $\widetilde{\mathsf Z}_{j,q}$ is sector-respecting and supported on $q$, and
$\mathsf Z_j=\sum_{q\in Q_j}\widetilde{\mathsf Z}_{j,q}$.
Discarding any zero $\widetilde{\mathsf Z}_{j,q}$ yields a sector-respecting decomposition with at most $|Q_j|$ nonzero atoms, proving
$m_j^{\min}\le |Q_j|$. The final statement follows from the definition of \textbf{P5-min}.
\end{proof}

\begin{definition}[Quantized index hypothesis {\normalfont\textsf{(ID: D-ECT-P4IDX-01)}}]\label{def:ect-p4-idx}
\mbox{}\par
Fix depth $j$ and a sectorization as in Definition~\ref{def:ect-sectorization}.
We say the sectorization has a \emph{quantized index} if there exists a constant $C_0\ge 1$ and an \emph{injective} map
\[
\iota_j: Q_j \hookrightarrow \{0,1,2,\dots,\lceil C_0(j+1)\rceil-1\}.
\]
Equivalently, each sector admits a distinct integer label from a set of size $O(j)$.
This is an explicit modeling/design hypothesis encoding bounded interface complexity; it is not implied by the other primitives and must be verified in each intended instantiation.
\end{definition}

\begin{lemma}[Quantized index implies linear sector count {\normalfont\textsf{(ID: L-ECT-QSIZE-01)}}]\label{lem:ect-qsize}
Under the quantized index hypothesis of Definition~\ref{def:ect-p4-idx}, we have
\[
|Q_j|\ \le\ \lceil C_0(j+1)\rceil.
\]
\end{lemma}
\begin{proof}
Since $\iota_j$ is injective, $|Q_j|\le \big|\{0,1,\dots,\lceil C_0(j+1)\rceil-1\}\big|=\lceil C_0(j+1)\rceil$.
\end{proof}

\begin{corollary}[Linear mode bound from P4 index {\normalfont\textsf{(ID: C-ECT-HL1-01)}}]\label{cor:ect-hl1}
\mbox{}\par
Assume a sectorization (Definition~\ref{def:ect-sectorization}) and P5-min packaging (Definition~\ref{def:ect-p5min}).
Assume also the quantized index hypothesis (Definition~\ref{def:ect-p4-idx}).
Then the packaged mode count satisfies
\[
m_j \ \le\ |Q_j|\ \le\ \lceil C_0(j+1)\rceil.
\]
\end{corollary}
\begin{proof}
The inequality $m_j\le |Q_j|$ follows from Lemma~\ref{lem:ect-mode-ub} under P5-min packaging.
The bound on $|Q_j|$ is Lemma~\ref{lem:ect-qsize}.
\end{proof}

\begin{remark}[Hard lemma slots behind coercivity emergence]\label{rem:ect-slots}
Lemma~\ref{lem:ect-mech} is purely algebraic. In applications, the substantive work is to justify:
(i) a \emph{sector/atom model} produced by packaging (P5) in which a uniform per-atom ICAP bound holds,
(ii) a \emph{mode-count compression} estimate $m_j\le C_0(j+1)$ (often a P4-style sectorization statement; cf.\ Lemma~\ref{lem:ect-mode-ub} and Corollary~\ref{cor:ect-hl1}),
and (iii) a feasibility/gating principle (P2) strong enough to rule out lossless feasible directions (if strict coercivity is required beyond ICAP).
\end{remark}

\paragraph{Audit monotonicity defect (optional).}
Let $\mathcal A$ be an audit functional equipped with a target monotonicity principle under coarse observation
(e.g.\ the data-processing inequality for path reversal asymmetry).
Given a refinement $f:Z\to X$ and a further coarsening $g:X\to Y$, define the monotonicity defect on an argument $(\cdot)$ by
\[
\mathrm{Def}_{\mathcal A}(g\mid f;\,\cdot)\ :=\ \max\!\bigl\{0,\ \mathcal A((g\circ f)_\#(\cdot))-\mathcal A(f_\#(\cdot))\bigr\}.
\tag{D-TK-DEF-01}\label{eq:tk-audit-def}
\]
When DPI holds in the stated setting, this defect vanishes identically.

\subsubsection*{Defect propagation rules (toolkit)}
The following inequalities are repeatedly useful.
They require no system-specific structure beyond norms and triangle inequalities.

\begin{enumerate}
  \item \textbf{Composition perturbation (triangle inequality).}
  For composable maps $A,B,A',B'$ between normed spaces with an operator norm $\|\cdot\|$,
  \[
  \|A\circ B - A'\circ B'\|\ \le\ \|A-A'\|\,\|B\| + \|A'\|\,\|B-B'\|.
  \]
  \item \textbf{Idempotence factorization.}
  For any map $E$ with a submultiplicative operator norm,
  \[
  E^2-E \;=\; E\circ(E-\mathrm{Id})\qquad\Rightarrow\qquad \|E^2-E\|\ \le\ \|E\|\,\|E-\mathrm{Id}\|.
  \]
  \item \textbf{Audit monotonicity as a zero-defect rule.}
  In contexts where $\mathcal A$ is known to contract under pushforward (as in Section~\ref{sec:aot}),
  $\mathrm{Def}_{\mathcal A}(g\mid f;\cdot)=0$ provides a standardized certificate that coarse observation does not create spurious directionality.
\end{enumerate}

\bibliographystyle{alpha}
\bibliography{refs}

\newcommand{\etalchar}[1]{$^{#1}$}
\begin{thebibliography}{DvdMS23}

\bibitem[AGH{\etalchar{+}}12]{AltanerEtAl2012}
Bernhard Altaner, Stefan Grosskinsky, Stephan Herminghaus, Lukas Katth{\"a}n,
  Marc Timme, and J{\"u}rgen Vollmer.
\newblock Network representations of non-equilibrium steady states: Cycle
  decompositions, symmetries and dominant paths.
\newblock {\em Physical Review E}, 85:041133, 2012.

\bibitem[Bor94]{Borceux1994}
Francis Borceux.
\newblock {\em Handbook of Categorical Algebra 1: Basic Category Theory}.
\newblock Cambridge University Press, 1994.

\bibitem[BSvdM25]{BauerSeifertVanderMeer2025}
M.~T. Bauer, U.~Seifert, and J.~van~der Meer.
\newblock Stroboscopic measurements in markov networks: exact generator
  reconstruction vs. thermodynamic inference.
\newblock {\em Journal of Physics A: Mathematical and Theoretical}, 58:125001,
  2025.

\bibitem[CK11]{CsiszarKorner2011}
Imre Csisz{\'a}r and J{\'a}nos K{\"o}rner.
\newblock {\em Information Theory: Coding Theorems for Discrete Memoryless
  Systems}.
\newblock Cambridge University Press, 2 edition, 2011.

\bibitem[Coh63]{Cohen1963}
Paul~J. Cohen.
\newblock The independence of the continuum hypothesis.
\newblock {\em Proceedings of the National Academy of Sciences of the United
  States of America}, 50(6):1143--1148, 1963.

\bibitem[CT06]{CoverThomas2006}
Thomas~M. Cover and Joy~A. Thomas.
\newblock {\em Elements of Information Theory}.
\newblock Wiley, 2 edition, 2006.

\bibitem[DCLP23]{DalCengioLecomtePolettini2023}
S.~Dal~Cengio, V.~Lecomte, and M.~Polettini.
\newblock Geometry of nonequilibrium reaction networks.
\newblock {\em Physical Review X}, 13:021040, 2023.

\bibitem[DP02]{DaveyPriestley2002}
B.~A. Davey and H.~A. Priestley.
\newblock {\em Introduction to Lattices and Order}.
\newblock Cambridge University Press, 2 edition, 2002.

\bibitem[DvdMS23]{DeguntherVanderMeerSeifert2023}
J.~Deg{\"u}nther, J.~van~der Meer, and U.~Seifert.
\newblock Coarse-grained entropy production in markov processes.
\newblock {\em arXiv preprint}, 2023.

\bibitem[Hyt15]{Hyttinen2015}
Tapani Hyttinen.
\newblock Counting measure and forking in finite models.
\newblock In {\AA}sa Hirvonen, Juha Juutinen, Juha Kontinen, Kimmo
  R{\"a}s{\"a}nen, and Jouko V{\"a}{\"a}n{\"a}nen, editors, {\em Logic Without
  Borders: Essays on Set Theory, Model Theory, Philosophical Logic and
  Philosophy of Mathematics}, volume~44 of {\em Trends in Logic}. De Gruyter,
  2015.

\bibitem[KADK19]{KhudaBukhshEtAl2019}
Wasiur~R. KhudaBukhsh, Arnab Auddy, Yann Disser, and Heinz Koeppl.
\newblock Approximate lumpability for markovian agent-based models using local
  symmetries.
\newblock {\em Journal of Applied Probability}, 56(3):647--671, 2019.

\bibitem[Kel79]{Kelly1979}
Frank~P. Kelly.
\newblock {\em Reversibility and Stochastic Networks}.
\newblock Cambridge University Press, 1979.

\bibitem[KS60]{KemenySnell1960}
John~G. Kemeny and J.~Laurie Snell.
\newblock {\em Finite Markov Chains}.
\newblock Van Nostrand, 1960.

\bibitem[Kun11]{Kunen2011}
Kenneth Kunen.
\newblock {\em Set Theory}.
\newblock College Publications, 2011.

\bibitem[Lib04]{Libkin2004}
Leonid Libkin.
\newblock {\em Elements of Finite Model Theory}.
\newblock Springer, 2004.

\bibitem[LS99]{LebowitzSpohn1999}
Joel~L. Lebowitz and Herbert Spohn.
\newblock A gallavotti--cohen type symmetry in the large deviation functional
  for stochastic dynamics.
\newblock {\em Journal of Statistical Physics}, 95:333--365, 1999.

\bibitem[Mer11]{Merhav2010}
Neri Merhav.
\newblock Data processing theorems and the second law of thermodynamics.
\newblock {\em IEEE Transactions on Information Theory}, 57(8):4926--4939,
  2011.

\bibitem[ML98]{MacLane1998}
Saunders Mac~Lane.
\newblock {\em Categories for the Working Mathematician}, volume~5 of {\em
  Graduate Texts in Mathematics}.
\newblock Springer, 2 edition, 1998.

\bibitem[MS11]{MacphersonSteinhorn2011}
Dugald Macpherson and Charles Steinhorn.
\newblock Definability in classes of finite structures.
\newblock In {\em Finite and Algorithmic Model Theory}, pages 140--176.
  Cambridge University Press, 2011.

\bibitem[NJ10]{NilssonJacobi2010}
Martin Nilsson~Jacobi.
\newblock A robust spectral method for finding lumpings and metastable states
  of non-reversible markov chains.
\newblock {\em Electronic Transactions on Numerical Analysis}, 37:296--306,
  2010.

\bibitem[Nor97]{Norris1997}
J.~R. Norris.
\newblock {\em Markov Chains}.
\newblock Cambridge University Press, 1997.

\bibitem[Pol12]{Polettini2012}
Matteo Polettini.
\newblock {\em Geometric and Combinatorial Aspects of Nonequilibrium
  Statistical Mechanics}.
\newblock Phd thesis, Universit{\`a} di Bologna, 2012.

\bibitem[PPRV10]{Puglisi2010}
A.~Puglisi, S.~Pigolotti, L.~Rondoni, and A.~Vulpiani.
\newblock Entropy production and coarse-graining in markov processes.
\newblock {\em Journal of Statistical Mechanics: Theory and Experiment},
  2010(05):P05015, 2010.

\bibitem[RP81]{RogersPitman1981}
L.~C.~G. Rogers and J.~W. Pitman.
\newblock Markov functions.
\newblock {\em The Annals of Probability}, 9(4), 1981.

\bibitem[Sch76]{Schnakenberg1976}
J{\"u}rgen Schnakenberg.
\newblock Network theory of microscopic and macroscopic behavior of master
  equation systems.
\newblock {\em Reviews of Modern Physics}, 48:571--585, 1976.

\bibitem[Sei12]{Seifert2012}
Udo Seifert.
\newblock Stochastic thermodynamics, fluctuation theorems and molecular
  machines.
\newblock {\em Reports on Progress in Physics}, 75(12):126001, 2012.

\bibitem[SN07]{SinitsynNemenman2007}
N.~A. Sinitsyn and I.~Nemenman.
\newblock The berry phase and the pump flux in stochastic chemical kinetics.
\newblock {\em Europhysics Letters}, 77(5):58001, 2007.

\bibitem[Wil72]{Willems1972}
J.~C. Willems.
\newblock Dissipative dynamical systems part i: General theory.
\newblock {\em Archive for Rational Mechanics and Analysis}, 45(5):321--351,
  1972.

\end{thebibliography}

\end{document}